\documentclass[11pt]{article}
\usepackage[margin=1in]{geometry}
\usepackage{times}
\usepackage[compact]{titlesec}
\usepackage{amsfonts}
\usepackage{amsmath}
\usepackage{amsthm}
\usepackage{amssymb}
\usepackage{latexsym}
\usepackage{epic}
\usepackage{epsfig}
\usepackage{hyperref}
\usepackage{verbatim}
\usepackage[justification=centering]{caption}
\usepackage{color}
\usepackage{xcolor}
\allowdisplaybreaks[1]
\usepackage{natbib}
\usepackage{setspace}
\usepackage{derivative}
\usepackage{cancel}
\usepackage{array}
\usepackage{tabu}
\usepackage{pdfpages}
\usepackage{multicol}
\usepackage{graphicx}
\usepackage{subcaption}
\usepackage{mathtools}
\usepackage[ruled]{algorithm2e}
\usepackage{multirow}
\usepackage{booktabs}

\usepackage{optidef}

\DeclareMathOperator*{\arginf}{\arg\,\inf}
\newcommand{\D}{\,\mathrm{d}}

\setcitestyle{authoryear}

\newtheorem{theorem}{Theorem} 

\newtheorem{lemma}{Lemma}
\newtheorem*{lemma*}{Lemma}

\newtheorem{assumption}{Assumption}
\newtheorem*{proposition*}{Proposition}

\newtheorem{definition}{Definition}
\newtheorem{remark}{Remark}

\newtheorem{example}[theorem]{Example}

\newtheorem*{conjecture*}{Conjecture}
\newtheoremstyle{nonindented}{1ex}{1ex}{}{}{\bfseries}{.}{.5em}{}
\newtheoremstyle{indented}{1ex}{1ex}{\itshape\addtolength{\leftskip}{0.6cm}\addtolength{\rightskip}{0.6cm}}{}{\bfseries}{.}{.5em}{}
\theoremstyle{nonindented}

\theoremstyle{indented}
\newtheorem*{direction*}{Research Direction}
\theoremstyle{plain}

\newcommand\blfootnote[1]{%
  \let\thefootnote\relax%
  \footnotetext{#1}%
  \let\thefootnote\svthefootnote%
}

\renewcommand{\hat}{\widehat}
\renewcommand{\tilde}{\widetilde}
\renewcommand{\bar}{\overline}

\def\min{\qopname\relax n{min}}
\def\max{\qopname\relax n{max}}
\def\argmin{\qopname\relax n{argmin}}
\def\argmax{\qopname\relax n{argmax}}

\newcommand{\eat}[1]{}

    \title{Mechanism Design via Market Clearing-Prices for Value Maximizers under Budget and RoS Constraints}
    
\author{Xiaodong Liu\\  Renmin University of China \\ \small xiaodong.liu@ruc.edu.cn 
        \and 
        Weiran Shen\\ Renmin University of China  \\ \small shenweiran@ruc.edu.cn
        \and 
        Zihe Wang\\  Renmin University of China \\ \small wang.zihe@ruc.edu.cn }
\begin{document}
 
\begin{titlepage}

\date{}

\maketitle
\thispagestyle{empty}    
    
    \begin{abstract}
    The transition to auto-bidding in online advertising has shifted the focus of auction theory from quasi-linear utility maximization to value maximization subject to financial constraints. We study mechanism design for buyers with private budgets and private Return-on-Spend (RoS) constraints, but public valuations, a setting motivated by modern advertising platforms where valuations are predicted via machine learning models.

We introduce the extended Eisenberg-Gale program, a convex optimization framework generalized to incorporate RoS constraints. We demonstrate that the solution to this program is unique and characterizes the market's competitive equilibrium. Based on this theoretical analysis, we design a market-clearing mechanism and prove two key properties: (1) it is incentive-compatible with respect to financial constraints, making truthful reporting the optimal strategy; and (2) it achieves a tight $\frac{1}{2}$-approximation of the first-best revenue benchmark, the maximum revenue of any feasible mechanism, regardless of IC. Finally, to enable practical implementation, we present a decentralized online algorithm.  Ignoring logarithmic factor, we prove that under this algorithm, both the seller’s revenue and each buyer’s utility converge to the equilibrium benchmarks with a sublinear regret of $\tilde{O}(\sqrt{m})$ over $m$ auctions.
    \end{abstract}

    %  A decision maker is deciding between an active action (e.g., purchase a house, invest certain stock) and a passive action.  The payoff of the active action depends on the buyer's private type and also an unknown state of nature. 
    %  An information seller can design experiments to reveal information about the realized state to the decision maker, and  would like to maximize profit from selling such information.  We  fully characterize, in closed-form, the revenue-optimal information selling mechanism for the seller.
    %     After eliciting the buyer's type, the optimal mechanism charges the buyer an upfront payment and then simply reveals whether the realized state passed a certain threshold or not. The optimal mechanism features both price discrimination and information discrimination. 
    %   The special buyer type who is a-priori indifferent between the active and passive action benefits the most from participating the mechanism.
       
\end{titlepage}
% \begin{frontmatter}

    % \runtitle{Optimal Pricing of Information}
    
%     \begin{aug}
%         % use \particle for den|der|de|van|von (only lc!)
%         % [id=?,addressref=?,corref]{\fnms{}~\snm{}\ead[label=e?]{}\thanksref{}}
%         %
%         %% e-mail is mandatory for each author
%         %
%         %%% initials in fnms (if any) with spaces
%         %
%         \author[id=au1,addressref={add1}]{\fnms{Shuze}~\snm{Liu}\ead[label=e1]{sl5nw@virginia.edu}}
%         \author[id=au2,addressref={add2}]{\fnms{Weiran}~\snm{Shen}\ead[label=e2]{shenweiran@ruc.edu.cn}}
%         \author[id=au3,addressref={add1}]{\fnms{Haifeng}~\snm{Xu}\ead[label=e3]{hx4ad@virginia.edu}}
%         %%%%%%%%%%%%%%%%%%%%%%%%%%%%%%%%%%%%%%%%%%%%%%
%         %% Addresses                                %%
%         %%%%%%%%%%%%%%%%%%%%%%%%%%%%%%%%%%%%%%%%%%%%%%
%         \address[id=add1]{%
% %            \orgdiv{Department of Computer Science},
%             \orgname{University of Virginia}}
        
%         \address[id=add2]{%
% %            \orgdiv{Department of the Second and Third Authors},
%             \orgname{Renmin University of China}}
%     \end{aug}
    
    %% Put support info here.  Reminder: do not thank the handling coeditor anonymously or by name

    % \begin{keyword}
    %     % \kwd{\color{red} Information Pricing}
    %     \kwd{Selling information}
    %     \kwd{mechanism design}
    %     \kwd{threshold experiments}
    %     \kwd{mixed virtual values}
    % \end{keyword}
    
% \end{frontmatter}

% \sz{
% econometrica requirements:
% 45 pages published 
% 25 online appendix

% %Related econometrica journal:
% %\url{https://onlinelibrary.wiley.com/doi/abs/10.3982/ECTA13251}
% %\url{https://onlinelibrary.wiley.com/doi/abs/10.3982/ECTA17260}
% }

 \section{Introduction}
 %Auction design, as a central topic at the intersection of economics and computer science, has attracted much attention since the seminal studies on the Vickrey-Clarke-Groves (VCG) auction model~\cite{groves1973incentives,vickrey1961counterspeculation,clarke1971multipart} and the optimal auction design for a single item~\cite{myerson1981optimal}. A wide range of auction mechanisms across diverse fields have been proposed in recent decades~\cite{cramton2004combinatorial,aggarwal2006truthful,varian2007position,edelman2007internet}. One of the most impactful applications of auction theory is online advertising, which has become the primary revenue source for numerous global Internet companies. Platforms such as Google, TikTok, and Meta rely heavily on auction-based mechanisms to allocate advertising slots and charge advertisers accordingly. 

The rise of auto-bidding systems has fundamentally transformed online advertising auctions. Traditional auction theory assumes that buyers are quasi-linear utility maximizers who directly submit bids for each auction. In practice, however, many advertisers delegate bidding decisions to algorithmic auto-bidders\footnote{See \url{https://support.google.com/google-ads/answer/2979071?hl=en} and \url{https://www.facebook.com/business/help/1619591734742116?id=2196356200683573}.}. These systems simplify the bidding process for advertisers by removing the need to set fine-grained bids for each auction. Instead, advertisers only need to submit their financial constraints, such as budget limits and target Return on Spend (RoS), to the auto-bidders. The auto-bidders then use algorithms to optimize bids for each auction, aiming to maximize the advertiser's total value across multiple auctions subject to financial constraints. In the auto-bidding systems, the buyers' valuations are not private information, because the platform can use historical data and machine learning models to predict valuations accurately, but buyers' budgets and target RoS are private information, which relies on the buyers' reporting.

% In an auto-bidding system, the advertisers first strategically report their financial constraints to the auto-bidders. Their valuations for each auction are considered as public information to the corresponding auto-bidder, as they can be accurately predicted via machine learning models based on historical data. The auto-bidders then optimize bids based on both the predicted valuations and the financial constraints reported by the advertisers. 

The change in buyers' models opens up opportunities to redesign mechanisms that are better suited for value maximizers, as traditional mechanisms, designed for utility maximizers, may not perform well for value maximizers. A growing body of literature addresses this topic. For instance, \citet{balseiro2021landscape} considers value maximizers subject to only the RoS constraint; \citet{balseiro2022optimal} designs a mechanism for value maximizers with a public budget constraint and a private RoS constraint; \citet{xing2023truthful} characterizes the mechanisms for value maximizers with a private budget constraint and a private RoS constraint, but does not provide an explicit and simple mechanism, while \citet{balseiro2023optimal} offers an optimal mechanism for a single value maximizer with both private valuations and a private RoS constraint.

% However, a general, explicit and simple mechanism for multiple value-maximizer with private budget and RoS constraints remains elusive.

In this paper, we focus on the problem of designing a mechanism for value maximizers with both private budget and RoS constraints. Rather than directly apply the revelation principle~\cite{myerson1979incentive,groves1973incentives}, we take a market perspective via a competitive equilibrium framework. In detail, we treat the auction ecosystem as a market where all buyers act as price-takers, maximizing value subject to their budget and RoS constraints. The platform's role is to compute market-clearing prices, where all buyers' total demand equals the platform's total supply. The competitive equilibrium framework offers several compelling advantages. First, the market-clearing price balances the total demand and the total supply, making them equal. Second, under market-clearing prices, the seller maximizes revenue, and all buyers maximize their own utilities.

Following the literature~\cite{aggarwal2024auto}, we consider buyers participating in multiple auctions simultaneously. This modeling approach simplifies the complexity of Bayesian modeling. Because numerous auctions run in online advertising systems, with the law of large numbers, these two modeling perspectives are asymptotically equivalent. This perspective also connects our problem to the Fisher market model and the Eisenberg-Gale program~\cite{eisenberg1959consensus}, which characterizes competitive equilibrium for value-maximizers with only budget constraints.

The standard competitive equilibrium framework does not directly extend to our setting due tothe presence of RoS constraints. With budget constraints alone, buyer demand decreases continuously as prices rise. If multiple buyers want the same item, raising the price smoothly reduces demand until supply equals demand. This continuity ensures the existence of equilibrium through standard fixed-point arguments. However, RoS constraints induce discontinuities in demand. A buyer purchases an item only if its value-to-price ratio meets their RoS threshold. At the critical price, the buyer's demand jumps discontinuously. They either demand one full unit or nothing. When multiple buyers have identical prices for the same item, traditional price adjustments fail:
\begin{itemize}
    \item Raising the price even slightly causes buyers to drop out (demand collapses to zero);
    \item Lowering the price maintains excess demand (multiple buyers want one item);
    \item No price can resolve the tie through the standard competitive equilibrium.
\end{itemize}

We resolve this through a modified competitive equilibrium framework. The platform can act as a centralized coordinator; when ties occur, the platform strategically allocates fractional supplies to tied buyers, ensuring global market clearance. 

To compute the modified competitive equilibrium, we extend the standard Eisenberg-Gale program to incorporate RoS constraints, named the extended Eisenberg-Gale program. We characterize the competitive equilibrium for value maximizers with budget and RoS constraints and construct the market-clearing mechanism based on the program's solution. While the centralized solution requires global knowledge of all auctions, we leverage the structure of the solution to interpret it as a first-price auction with uniform bidding strategies. This insight bridges the gap between market-clearing mechanisms and auto-bidding in first-price auctions~\cite{alimohammadi2023incentive,conitzer2022pacing,deng2025no,aggarwal2025no}. Consequently, we propose an online implementation of our mechanism using the Regularized Dual Averaging (RDA) algorithm~\cite{xiao2010dual}. Unlike previous works~\cite{deng2025no,aggarwal2025no} that treat other buyers' highest bids as a static distribution, our approach takes an equilibrium perspective. Our algorithm is decentralized and parameter-free: each buyer updates their strategy based solely on binary feedback (whether he wins the auction or not).  Ignoring logarithmic factor, we prove that our algorithm converges with a utility regret of $\tilde{O}(\sqrt{m})$ over $m$ auctions, significantly improving upon the $\tilde{O}(m^{3/4})$ bound under bandit feedback in this setting.

\subsection{Main Results}
We summarize our contributions as follows:
\begin{enumerate}
    \item The standard Eisenberg-Gale convex program models buyers with budget constraints only. We extend this program to accommodate buyers with both budget and RoS constraints. Under mild assumptions, we characterize the solution of this extended program as a competitive equilibrium, showing its existence, uniqueness, and Pareto efficiency (See Theorem \ref{thm:competitive equilibrium characterization}).
    \item We construct the market-clearing mechanism based on the outcome of the extended Eisenberg-Gale program, analyzing its IC properties (See Theorem \ref{thm:ic}) and the revenue approximation compared to the first-best revenue, the maximum revenue achieved for any feasible allocation without IC constraints (See Theorem \ref{thm:revenue_approximation}).

    \item We formulate the dual program of the extended Eisenberg-Gale program as an online learning problem. The dual objective consists of both a price term and a strongly convex regularizer. We apply the regularized dual averaging (RDA) algorithm~\cite{xiao2010dual} in our setting and propose an online implementation of the market-clearing mechanism (See Algorithm~\ref{alg:RDA}). Algorithm~\ref{alg:RDA} requires no pre-set hyperparameters. Due to the structure of the dual program, Algorithm~\ref{alg:RDA} can operate in a decentralized manner, with each buyer updating their strategy based on binary feedback (whether he wins the auction or not). We further prove that the seller's revenue over $m$ auctions is at least the optimal offline revenue plus a sublinear regret of $\tilde{O}(\sqrt{m})$ and each advertiser's total utility over $m$ auctions suffers a sublinear regret of $\tilde{O}(\sqrt{m})$ (See Theorem \ref{thm:alg}).

\end{enumerate}

% \subsection{Paper Organization}
% The remainder of this paper is organized as follows.

% \textbf{Section~2} introduces the model setup, including the buyer utility model with budget and RoS constraints, the competitive equilibrium definition, and the market-clearing price.

% \textbf{Section~3} begins with reviewing the classical Eisenberg-Gale convex program and then presents our extended Eisenberg-Gale convex program that incorporates RoS constraints. We establish the correspondence between optimal solutions to our program and competitive equilibria (Theorem~\ref{thm:competitive equilibrium characterization}), prove uniqueness and Pareto efficiency, and show incentive compatibility (Theorem~\ref{thm:ic}). We then analyze the platform's revenue, proving a tight $\frac{1}{2}$-approximation to the first-best benchmark (Theorem~\ref{thm:revenue_approximation}).

% \textbf{Section~4} develops our online learning algorithm for sequential auctions. We present a decentralized algorithm based on Regularized Dual Averaging and prove convergence guarantees: $\tilde{O}\left(\sqrt{m }\right)$ regret for each buyer's utility and the seller's revenue (Theorem~\ref{thm:alg}). This section demonstrates how the market-clearing mechanism can be implemented in practice without requiring global information.

% All formal proofs are provided in the Appendix, along with auxiliary lemmas and technical details.

\subsection{Related Work}
Our work lies at the intersection of mechanism design for auto-bidding, Fisher market theory, and online learning. We discuss the most closely related work in each area below.

\textbf{Mechanism design in auto-bidding auctions.} The design of optimal mechanisms for advertisers with financial constraints has been studied extensively. \citet{balseiro2021landscape} characterize optimal mechanisms for value maximizers and utility maximizers under RoS constraints across various information structures. \citet{balseiro2022optimal} propose optimal mechanisms for value maximizers with public valuations, public budgets, and private RoS constraints. \citet{xing2023truthful} characterize feasible mechanisms for value maximizers with public valuations and private budget and RoS constraints. \citet{balseiro2023optimal} study optimal mechanisms for a single buyer with private valuation and RoS constraint.

Our work takes a different approach: rather than optimizing for revenue directly, we characterize the market-clearing mechanism and analyze its incentive and revenue properties. Our mechanism is not optimal in general but has the advantage of being simple, implementable via standard first-price auctions, and amenable to online learning.

\textbf{Auction analysis in auto-bidding auctions} Another line of work studies the properties of traditional auctions for buyers with financial constraints. \citet{aggarwal2019autobidding} formulate the problem mathematically and prove that uniform bidding is optimal in truthful auctions. They also prove the Price of Anarchy (PoA) is $\frac{1}{2}$ when buyers have financial constraints in truthful auctions. \citet{deng2022efficiency} and \citet{liaw2024efficiency} extend to study the PoA of first-price auctions for buyers having financial constraints. \citet{feng2024strategic} consider scenarios where advertisers' objectives differ from auto-bidders', showing that misreporting can be beneficial, and showing that the PoA is $\frac{1}{2}$ at any pure Nash equilibrium and $\frac{1}{4}$ at any mixed Nash equilibrium.

Most relevant to our work, \citet{alimohammadi2023incentive} prove that first-price auctions with uniform-bidding strategies are auto-bidding incentive compatible. We provide an extended Eisenberg-Gale program, showing that the solution of this program is the equilibrium of first-price auctions. We also give a proof of incentive compatibility of first price auctions from a market perspective.

\textbf{Fisher markets and Eisenberg-Gale program.} Our work builds on the Fisher market model and the Eisenberg-Gale convex program~\citep{eisenberg1959consensus}, which computes market equilibrium for buyers with budget constraints. Several extensions have been proposed: \citet{bei2017earning} introduce earning limits for sellers; \citet{birnbaum2010new,devanur2004spending,vazirani2010spending} extend the framework to price-dependent utilities; \citet{devanur2016new} considers utility upper bounds; \citet{jalota2023fisher} introduces knapsack constraints.

Most closely related, \citet{conitzer2022pacing} use the Eisenberg-Gale program to compute first-price pacing equilibria for buyers with quasi-linear utilities under budget constraints. However, their solution concept differs from ours: their pacing equilibrium does not require buyers to maximize utilities individually. Furthermore, they do not consider RoS constraints, which fundamentally alter the buyer model (from quasi-linear to value maximizers). Our equilibrium analysis is substantially different from theirs.

\textbf{Online learning in Fisher markets.} Online variants of Fisher markets have been studied in settings where goods arrive sequentially~\citep{azar2010allocate,banerjee2022online,gao2021online} or users arrive sequentially~\citep{jalota2024stochastic}. Our work is most related to \citet{gao2021online}, who analyze online Fisher markets with budget constraints using dual averaging. The key distinction is our inclusion of RoS constraints. Also, the approach used to analyze each buyer's utility is different. For value maximizers with additional RoS constraints, their method cannot be applied to our setting.

\textbf{Online learning in first price auctions.} Prior work studies learning to bid in first-price auctions with financial constraints. \citet{ai2022no,castiglioni2022online,han2024optimal,wang2023learning} focus on budget-constrained buyers. \citet{aggarwal2025no,castiglioni2024online,liang2024online,deng2025no} study budget and RoS constraints. However, these studies typically analyze a single advertiser learning algorithm against a fixed environment, often using bandit algorithms. 
In contrast, our algorithm is built on the dual program of the extended Eisenberg-Gale program. We prove that when each buyer uses our proposed algorithm, each buyer's strategy converges to the optimal strategy.
Combining these results, we show that under bandit feedback, our algorithm guarantees that each buyer's utility suffers at most $\tilde{O}\left(\sqrt{m} \right)$. A key improvement of our work over \citet{aggarwal2025no} and \citet{deng2025no} lies in the information structure required for learning. In their frameworks, a buyer must learn the distribution of the highest competing bid to determine their optimal strategy. Conversely, our algorithm establishes that when buyers update their strategies simultaneously, this external information is unnecessary. We show that local feedback—specifically, a buyer's own realized utility—is sufficient to guide the market to equilibrium, effectively decoupling the buyers' learning processes.

Our algorithm is simple and decentralized. Each buyer's strategy update algorithm only depends on their own valuation and binary feedback (whether he wins the auction or not). Moreover, the algorithm is parameter-free. 

\section{Preliminaries}
We consider an online advertising platform (the seller) aiming to sell a set of items (advertising slots)  to multiple buyers (advertisers). Let $[m] = \{1, 2, \dots, m\}$ and $[n] = \{1, 2, \dots, n\}$ denote the set of items and the set of buyers, respectively. The items are sold through $m$ auctions, where each auction sells one item.
\subsection{Valuations and Financial Constraints}
Each buyer $i \in [n]$ has a valuation $v_{i, j} \in \mathcal{V} = [0, \bar{v}]$ for item $j \in [m]$, which represents the maximum amount buyer $i$ is willing to pay for item $j$. Suppose that for every item $j$, there exists at least one buyer with strictly positive valuation, i.e., $\max_i \{v_{i,j}\} > 0$, since otherwise the item can be safely ignored. Denote by $v_{\cdot, j} = (v_{1, j}, \cdots, v_{n,j})$ the buyers' value profile for item $j$. We assume the valuation profile $v = (v_{i,j})_{i \in [n], j \in [m]}$ is public information to the platform. This assumption is standard in the auto-bidding literature, justified by the platform's ability to predict valuations using high-volume historical data accurately. While valuations are public, the buyers face private financial constraints. Specifically, the private information of buyer $i$ is characterized by a type $\theta_i = (\lambda_i, \tau_i)$, where:
\begin{itemize}
    \item $\lambda_i \in \Lambda \subseteq \mathbb{R}_+$ is buyer $i$'s budget, which bounds the total payment buyer $i$ can incur across all $m$ auctions;
    \item $\tau_i \in \mathcal{T} \subseteq [1, \infty)$ is buyer $i$'s target Return-on-Spend, which bounds the efficiency of spend, requiring the total value obtained to be at least $\tau_i$ times the total payment.
\end{itemize}
Let $\lambda = (\lambda_1, \dots, \lambda_n)$ be the budget constraint profile and $\tau = (\tau_1, \dots, \tau_n)$ the RoS constraint profile. For simplicity, we denote the profile of constraints excluding buyer $i$ by $\lambda_{-i} = (\lambda_1, \cdots, \lambda_{i-1}, \lambda_{i+1},\cdots,  \lambda_n)$ and $\tau_{-i} = (\tau_1, \cdots, \tau_{i-1}, \tau_{i+1}, \cdots, \tau_{n})$, respectively.
\subsection{Mechanism}
Let $\hat{\lambda} = (\hat{\lambda}_1, \dots, \hat{\lambda}_n)$ and $\hat{\tau} = (\hat{\tau}_1, \dots, \hat{\tau}_n)$ denote the reported constraint profiles of the buyers, respectively. A mechanism consists of both an allocation rule $a$ and a payment rule $t$.
\begin{definition}[Mechanism]
A mechanism is a tuple $(a, t)$, where:
\begin{itemize}
    \item $a(v, \hat{\lambda}, \hat{\tau}) = \{a_{i,j}(v_{\cdot, j}, \hat{\lambda}, \hat{\tau})\}_{i \in [n], j \in [m]}$ is the allocation rule, where $a_{i,j}(v_{\cdot, j}, \hat{\lambda}, \hat{\tau}):(\mathcal{V} \times \Lambda \times \mathcal{T})^n \mapsto [0, 1]^n$ denotes the probability (or fraction) of item $j$ allocated to buyer $i$. Feasibility requires $\sum_{i \in [n]} a_{i,j}(v_{\cdot, j}, \hat{\lambda}, \hat{\tau}) \le 1$ for all $j$.
    \item $t(v, \hat{\lambda}, \hat{\tau}) = \{t_{i,j}(v_{\cdot, j}, \hat{\lambda}, \hat{\tau})\}_{i \in [n], j \in [m]}$ is the payment rule, where $t_{i,j}(v_{\cdot, j}, \hat{\lambda}, \hat{\tau}): (\mathcal{V} \times \Lambda \times \mathcal{T})^n \mapsto  \mathbb{R}_+$ denotes the payment collected from buyer $i$ for item $j$.
    \end{itemize}
\end{definition}
% Note that since valuations $v$ are public, the mechanism may depend on them implicitly; we omit $v$ from the notation for brevity.
\subsection{Buyer Utility, Individual Rationality, and Incentive Compatibility}
%Buyers are value maximizers subject to budget and RoS constraints. Now suppose buyers' reports are $(\hat{\lambda},  \hat{\tau})$, while buyers' private constraints are $(\lambda, \tau)$.
Given buyers' reported constraints $(\hat{\lambda}, \hat{\tau})$, the budget constraint of buyer $i$ requires that their total payment does not exceed the budget $\lambda_i$:
\begin{align}
    \label{eq:budget constraint}
	\sum_{j=1}^m t_{i, j}(v_{\cdot, j}, \hat{\lambda}, \hat{\tau}) \le \lambda_i.
\end{align}
The RoS constraint of buyer $i$ is satisfied if the total value obtained from the auctions is at least $\tau_i$ times the total payment:
\begin{align}
    \label{eq:ros constraint}
	\sum_{j=1}^m a_{i, j}(v_{\cdot, j}, \hat{\lambda}, \hat{\tau})v_{i,j} \ge \tau_i \sum_{j=1}^m t_{i, j}(v_{\cdot, j}, \hat{\lambda}, \hat{\tau}).	
\end{align}
The buyers are value maximizers subject to the above budget and RoS constraints, i.e., the utility of buyer $i$ is the total realized value if the constraints are satisfied, and $-\infty$ otherwise. Formally:
\begin{align}
    u_i(\hat{\lambda}, \hat{\tau}; \lambda_i, \tau_i) = 
    \begin{cases}
        \sum_{j=1}^m a_{i,j}(v_{\cdot,j}, \hat{\lambda}, \hat{\tau}) v_{i,j} & \text{if constraints \eqref{eq:budget constraint} and \eqref{eq:ros constraint} are satisfied,} \\
        -\infty &\text{otherwise.}
    \end{cases}
    \label{eq:utility}
\end{align}
Note that a buyer's utility drops to $-\infty$ if the outcome violates their \textit{true} constraints, rather than their reported ones.

The mechanism proceeds as follows:
\begin{enumerate}
    \item Each buyer privately observes their true constraints $(\lambda_i, \tau_i)$;
    \item Each buyer reports constraints $(\hat{\lambda}_i, \hat{\tau}_i)$ to the platform;
    \item The platform computes allocations $x(v, \hat{\lambda}, \hat{\tau})$ and payments $t(v, \hat{\lambda}, \hat{\tau})$ according to the mechanism rules;
    \item Each buyer obtains utility $u_i(\hat{\lambda},\hat{\tau}; \lambda_i, \tau_i)$.
\end{enumerate}

We consider mechanisms that are incentive compatible (IC) and individually rational (IR).

\begin{definition}[Incentive Compatibility]
    A mechanism is incentive compatible (or truthful), if for any buyer $i$, any true type $(\lambda_i, \tau_i)$, any reported type $(\hat{\lambda}_i, \hat{\tau}_i)$, and any type profile $(\lambda_{-i}, \tau_{-i})$ of other buyers, we have:
    \begin{equation*}
        u_i(\lambda_i, \lambda_{-i},\tau_{i}, \tau_{-i}; \lambda_{i}, \tau_{i}) \ge u_i(\hat{\lambda}_i, \lambda_{-i}, \hat{\tau}_i, \tau_{-i}; \lambda_{i}, \tau_{i}).
    \end{equation*}
\end{definition}

\begin{definition}[Individual Rationality]
    A mechanism is individually rational if for any buyer $i\in[n]$, with true type profile $(\lambda, \tau)$, reporting truthfully guarantees non-negative utility regardless of other buyers' reports:
    \begin{gather*}
        u_i(\lambda_i, \hat{\lambda}_{-i}, \tau_i, \hat{\tau}_{-i}; \lambda_i, \tau_i) \ge 0, \quad \forall \hat{\lambda}_{-i}, \hat{\tau}_{-i}.
    \end{gather*}
\end{definition}

Note that according to Equation \eqref{eq:utility}, a buyer's utility is non-negative if both the budget and RoS constraints are satisfied. 

\subsection{Competitive Equilibrium and Market-Clearing Price}

Rather than directly designing allocation and payment rules via the revelation principle~\citep{myerson1979incentive}, we adopt a \emph{market-based perspective} inspired by classical equilibrium theory. We reformulate the mechanism design problem as a market where:
\begin{itemize}
    \item Each item $j$ has a unit price $p_j$;
    \item Buyers act as price-takers, optimizing their demand given prices and constraints;
    \item The seller coordinates supply to ensure market clearance.
\end{itemize}

This perspective offers several advantages. First, it allows us to leverage powerful convex optimization techniques to compute equilibrium outcomes. Second, market-clearing prices naturally balance supply and demand, leading to economically interpretable mechanisms. Third, as we show in Section~\ref{sec:theory}, the resulting allocation can be implemented via standard first-price auctions with uniform bidding strategies.

In standard competitive equilibrium theory, equilibrium consists of prices and allocations where (i) each agent optimizes given prices, and (ii) markets clear (supply equals demand). However, as we discuss below, RoS constraints introduce discontinuities in buyer demand functions that can prevent the existence of standard competitive equilibria. We address this by introducing the concept of a \emph{modified competitive equilibrium}, which empowers the seller to coordinate supply allocation.

\begin{definition}[Modified Competitive Equilibrium]\label{def:modified_competitive_equilibrium}
    A modified competitive equilibrium consists of a price vector $p = (p_1, \ldots, p_m)$, buyer demands $\{x_{i,j}\}_{i \in [n], j \in [m]}$, and seller supplies $\{y_{i,j}\}_{i\in[n],j \in [m]}$ satisfying:
\begin{itemize}
    \item \textbf{Buyer Utility Maximization:} Each buyer $i$ maximizes utility given prices $p$ and their financial constraints:
    \begin{maxi}|m|[2]{x_{i, \cdot}}{\sum_{j=1}^m x_{i,j} v_{i,j}}{\tag{I-PRIMAL}\label{IOP:primal}}{}
    \addConstraint{\sum_{j=1}^m p_j x_{i,j}}{\le \lambda_i\quad}{\text{(Budget)}}
    \addConstraint{\sum_{j=1}^m x_{i,j}v_{i,j}}{\ge \tau_i\sum_{j=1}^m p_j x_{i,j}\quad}{\text{(RoS)}}
    \addConstraint{x_{i,j}}{\in [0, q_{i,j}]\quad}{\forall j \in [m],}
    \end{maxi}
    where $q_{i,j}$ is defined as:
    \begin{gather*}
        q_{i,j} = \begin{cases}
            y_{i,j} & \text{if } 0 < y_{i,j} < 1 \\
            1 &\text{otherwise.}
        \end{cases}
    \end{gather*}
    The buyer can demand up to the seller's allocated supply when the supply is fractional.
    \item \textbf{Seller Revenue Maximization:} The seller chooses supply $\left\{y_{i,j} \right\}_{i\in[n],j\in[m]}$ to maximize revenue:
    \begin{maxi*}|m|[2]{y}{\sum_{j=1}^m \sum_{i=1}^n p_{j}y_{i,j}}{}{}
    \addConstraint{\sum_{i=1}^n y_{i,j}}{\le 1}{\quad \forall j \in [m].}
    \end{maxi*}

    \item \textbf{Market Clearance:} Supply equals demand for each item:
    \begin{gather*}
        \sum_{i=1}^n x_{i, j}=\sum_{i=1}^n y_{i,j}, \quad \forall j\in [m].
    \end{gather*}
\end{itemize}
\end{definition}

\begin{definition}[Market-Clearing Price]
    A price vector $p = (p_1, \ldots, p_m)$ is a \emph{market-clearing price} if there exist demands $x$ and supplies $y$ such that $(p, x, y)$ forms a competitive equilibrium.
\end{definition}

\begin{remark}
    Although the seller is indeed a monopolist, we adopt competitive equilibrium as our solution concept because it accurately reflects the operational reality of modern advertising platforms. These platforms act as algorithmic market makers, computing clearing prices to balance supply and demand rather than strategically manipulating prices. This approach yields mechanisms that are simple, transparent, and amenable to decentralized implementation.

    The modification is necessitated by RoS constraints, which induce discontinuities in buyer demand functions. In markets with only budget constraints, fractional allocations naturally arise as prices adjust—higher prices reduce demand continuously until supply equals demand. However, with RoS constraints, a buyer's demand can drop discontinuously from 1 to 0 when the price crosses their RoS threshold $\frac{v_{i,j}}{\tau_i}$. This creates situations where multiple buyers simultaneously demand an item at a given price, yet no price adjustment can resolve the tie without eliminating all demand. Our modification empowers the seller to act as a centralized coordinator: when ties occur, the seller strategically allocates fractional supplies to ensure global market clearance. Crucially, buyers are indifferent among all feasible demands at the equilibrium price (as their RoS constraint binds), so accepting the seller's allocated supply is optimal. We formalize this in Section~\ref{sec:theory}.

\end{remark}

Since the seller incurs no production cost, their optimal strategy is to supply the full quantity of each item ($\sum_{i=1}^ny_{i,j} = 1$) at any non-negative price. Therefore, our mechanism design problem reduces to finding a price vector $p$ and allocation $x$ such that: (1) each buyer optimally responds to $p$ given their constraints, and (2) total demand equals total supply for each item.

\section{The Extended Eisenberg-Gale Market and the Market-Clearing Mechanism}
\label{sec:theory}

We begin our analysis by relaxing the incentive compatibility constraints, assuming that each buyer $i$ reports their private type $(\lambda_i, \tau_i)$ truthfully. This assumption enables us to focus on computing the market-clearing price $p$ and the associated allocation $x$. We define a mechanism $(a, t)$ that produces the same outcome of the competitive equilibrium given the profiles $(\lambda, \tau)$ as input. We then verify that the mechanism is indeed incentive compatible, ensuring truthful reporting.

In this section, we develop a computational framework for finding market-clearing prices and allocations. Our approach builds on the classical Eisenberg-Gale convex program, which characterizes competitive equilibria in markets with value-maximizing buyers subject to only budget constraints~\citep{eisenberg1959consensus}. However, the standard Eisenberg-Gale framework cannot be directly applied to our setting because our buyers are subject to both budget and RoS constraints. A critical technical challenge is that RoS constraints are defined over their valuations and payments. In the standard Eisenberg-Gale program, the decision variables are allocations, while prices emerge only as dual variables. 

We resolve this by introducing an \emph{extended Eisenberg-Gale convex program} that explicitly handles RoS constraints and show that optimal solutions to these programs precisely characterize competitive equilibria in our setting.

\subsection{The Extended Eisenberg-Gale Convex Program}
We start with the standard Eisenberg-Gale convex program. The program features two sets of variables:$\{x_{i,j}\}_{i\in [n],j\in [m]}$ and $\{c_i\}_{i \in [n]}$, where $x$ is the allocation and $c_i$ is the utility of buyer $i$:
\begin{maxi}|l|[2]{x,c}{\sum_{i=1}^n \lambda_i\log(c_i)}{\label{prog:standard}}{}
\addConstraint{c_i}{= \sum_{j=1}^m x_{i,j}v_{i,j}}{\quad\forall i \in [n]}
\addConstraint{\sum_{i=1}^nx_{i,j}}{\le 1}{\quad\forall j \in [m]}
\addConstraint{x_{i,j}}{\ge 0}{\quad\forall i \in [n], \forall j \in [m]}.
\end{maxi}
As described above, the standard Eisenberg-Gale convex program can only handle the budget constraints. To incorporate the RoS constraints, we need to build connections between the utility and the payment of each buyer. Note that the dual program of Program \eqref{prog:standard} involves the price of each item, which can be used to represent the payments of the buyers. Therefore, the intuition behind our extended program is to add RoS constraints to the dual program.

Introducing dual variables $\{w_i\}_{i\in[n]}$ and $\{ p_j \}_{j\in[m]}$, the dual program of Program \eqref{prog:standard} can be written as follows:
 \begin{mini*}|l|[2]{p, w}{\sum_{j=1}^m p_j - \sum_{i=1}^n \lambda_i \log(w_i) + \sum_{i=1}^n \lambda_i \left(\log(\lambda_i) - 1\right)}{}{}
\addConstraint{p_j}{\ge w_i v_{i,j}}{\quad \forall i \in [n], \forall j \in [m]}
\addConstraint{p_j}{\ge 0}{\quad \forall j\in [m]}
\addConstraint{w_i}{\ge 0}{\quad \forall i \in [n]},
\end{mini*}
where the variable $p_j$ can be interpreted as the price of item $j$.
% We now list the complementary slackness conditions below, which will be useful for us to propose the extended Eisenberg-Gale convex program.

Let $(x^*, c^*)$ and $(w^*,p^*)$ be optimal solutions to the above two programs.
% 去掉 (5), (6) 保留(7)但是去掉标号
The following complementary slackness condition holds: $x^*_{i,j}(p^*_j -w^*_i v_{i,j}) = 0,\forall i \in [n],j\in[m]$.
Summing over all $m$ items on both sides of the above equation gives: $ \sum_{j=1}^m p_j^*x_{i,j}^* = w_i^* \sum_{j=1}^mv_{i,j}x_{i,j}^*, \forall i \in [n].$
Notice that the summations on two sides of the above equation are the total payment and the utility of buyer $i$, respectively. This immediately implies that the actual RoS of buyer $i$ is $\frac{1}{w^*_i}$. Consequently, to deal with buyer $i$'s RoS constraint, we can add a new constraint $\frac{1}{w_i}\ge \tau_i$, or equivalently $w_i \le \frac{1}{\tau_i}$, to the dual program:
% If we add a new constraint on $w_i$ in the dual program, $w_i \le \frac{1}{\tau_i}$, the solution ensures that buyer $i$'s RoS constraint is satisfied.
% Therefore, the dual program of the extended Eisenberg-Gale convex program is as follows.
 \begin{mini}|l|[2]{p, w}{\sum_{j=1}^m p_j - \sum_{i=1}^n \lambda_i \log(w_i) + \sum_{i=1}^n \lambda_i \left(\log(\lambda_i) - 1\right)}{\tag{C-DUAL}\label{prog: dual}}{}
\addConstraint{p_j}{\ge w_i v_{i,j}}{\quad \forall i \in [n], \forall j \in [m]}
\addConstraint{w_i}{\le \frac{1}{\tau_i}}{\quad \forall i\in [n]}
\addConstraint{p_j}{\ge 0}{\quad \forall j\in [m]}
\addConstraint{w_i}{\ge 0}{\quad \forall i \in [n]},
\end{mini}
% Observe that in the above program, the last term of the objective function is constant, and thus can be removed safely. However, we still keep the term for now, as we will make use of the strong duality condition later.

Taking the dual of the above program again, we should obtain a new primal program, which we call the extended Eisenberg-Gale convex program:
% Introducing new variables $\{d_i\}_{i\in[n]}$, which is the dual variables of the new constraints, the extended Eisenberg-Gale convex program can be written as 
\begin{maxi}|l|[2]{x,d,c}{\sum_{i=1}^n \left[\lambda_i\log(c_i)-\frac{d_i}{\tau_i}\right]}{\tag{C-PRIMAL}\label{prog: primal}}{}
\addConstraint{c_i}{= \sum_{j=1}^m x_{i,j}v_{i,j} + d_i}{\quad\forall i \in [n]}
\addConstraint{\sum_{i=1}^nx_{i,j}}{\le 1}{\quad\forall j \in [m]}
\addConstraint{d_i}{\ge 0}{\quad\forall i \in [n]}
\addConstraint{x_{i,j}}{\ge 0}{\quad\forall i \in [n], \forall j \in [m]}.
\end{maxi}
Since we introduced a new set of constraints to its dual, our extended Eisenberg-Gale convex program contains a new set of variables: $\{d_i\}_{i\in[n]}$.

% Then omitting the constant terms in the dual program, the dual program can be written as
%  \begin{mini}|l|[2]{p, w}{\sum_{j=1}^m p_j - \sum_{i=1}^n \lambda_i \log(w_i)}{\tag{C-DUAL}\label{prog: dual}}{}
% \addConstraint{p_j}{\ge w_i v_{i,j}}{\quad \forall i \in [n], \forall j \in [m]}
% \addConstraint{w_i}{\le \frac{1}{\tau_i}}{\quad \forall i\in [n]}
% \addConstraint{p_j}{\ge 0}{\quad \forall j\in [m]}
% \addConstraint{w_i}{\ge 0}{\quad \forall i \in [n]},
% \end{mini}

% We first focus on the allocation problem associated with the mechanism. Note that our model is similar to the classic Fisher market model. However, the classical model does not address buyers' RoS constraints. 
With the above programs, we compute an allocation and pricing scheme that satisfies both the budget and RoS constraints for each buyer. And since these programs determine the outcome in a centralized way, we also call our extended program the central program.
% where $w_i$ and $p_j$ are the dual variables of the first two constraints in Program \eqref{prog: primal}, respectively.
% constraint $w_i = \sum_{j=1}^m x_i^j v_i^j + t_i$ and $p^j$ is the dual variable of the constraint $\sum_{i=1}^n x_i^j \le 1$.

 %Based on these solutions, 
% Viewing $x^*$ as allocations and $p^{*}$ as the prices of the items, 
% we can define a mechanism $\hat{\mathcal{M}}$ as follows:
% \begin{gather*}
%     \hat{x}_{i,j}=x^*_{i,j},\forall i,j\quad \text{and}\quad \hat{t}_{i}=\sum_{j}p^*_{j}x^*_{i,j},\forall i.
% \end{gather*}
% In other words, mechanism $\hat{\mathcal{M}}$ allocates the items according to $x^*$ and charges each buyer the total price of the allocated items. %It is worth noticing that mechanism $\hat{\mathcal{M}}$ always allocates exactly 1 unit of each item, i.e., $\sum_{i=1}^nx^*_{i,j}=1,\forall j$. 

% Since $x^*$ is the allocation and $p^*$ contains the prices of the items, the payment of buyer $i$ 

In the remainder of this section, we show that if we view the auctions as a market, a set of prices $p$ and an allocation $x$ form a competitive equilibrium if and only if they come from optimal solutions to the above programs.

\subsection{Competitive Equilibrium and Market-Clearing Mechanism}
In this subsection, we establish our first main result (Theorem~\ref{thm:competitive equilibrium characterization}): optimal solutions to Programs~\eqref{prog: primal} and~\eqref{prog: dual} precisely correspond to competitive equilibria. In particular:
\begin{itemize}
    \item Any optimal solution $(p^*, x^*, y^*)$ forms a competitive equilibrium (Lemma~\ref{lemma:competitive equilibrium});
    \item Conversely, any competitive equilibrium corresponds to optimal solutions of the programs (Lemma~\ref{lemma:IOP_SOP});
    \item The equilibrium is unique and Pareto efficient (Lemma~\ref{lemma:unique} and Lemma~\ref{lemma:Pareto_efficiency}).
\end{itemize}

Before starting our analysis, we introduce two technical assumptions to streamline the exposition.

% \begin{assumption}[Non-Atomic Valuations]\label{assmption:no_tie_breaking_rule}
%     For any item $j \in [m]$, the valuation $v_{i,j}$ of each buyer $i$ is drawn independently from a continuous (non-atomic) distribution.
% \end{assumption}

% This assumption is standard in the auction literature to preclude the need for complex tie-breaking rules~\cite{myerson1981optimal,balseiro2021landscape,balseiro2022optimal}. In the context of online advertising, valuations are typically products of cost per action, predicted click-through rates (pCTR), and conversion rates (pCVR). Since these underlying probabilities vary continuously across heterogeneous users and ad creatives, the probability that two distinct buyers have the same valuation for one item is zero.

\begin{assumption}[Competitive Market]\label{assumption:no_buyer_win_all}
    No single buyer $i \in [n]$ can purchase the entire set of items $[m]$ without violating their budget $\lambda_i$ or RoS constraint $\tau_i$.
\end{assumption}

This assumption reflects the competitive nature (or ``thickness") of real-world advertising markets. Given the vast number of auctions $m$ relative to any single advertiser's budget, it is structurally impossible for one buyer to monopolize the supply. \footnote{This ensures that the competitive equilibrium is driven by market-wide competition rather than the constraints of a single dominant agent.}

We now proceed with the formal analysis.
It is easy to verify that both of the above programs ((\ref{prog: primal}) and (\ref{prog: dual})) are convex. Let $(x^*,d^*,c^*)$ and $(w^*, p^*)$ be optimal solutions to Program \eqref{prog: primal} and \eqref{prog: dual}, respectively. We now list the complementary slackness conditions below, which will be useful for later arguments.
% Before starting our analysis, we first discuss some useful properties of the above two programs. 
% Since $(x^*,d^*,c^*)$ and $(w^*, p^*)$ are optimal solutions, they must satisfy the following complementary slackness conditions:
\begin{align}
	&w^*_i \left(c^*_i - \sum_{j=1}^m x^*_{i,j}v_{i,j} - d_i \right) = 0,&\forall i \in [n], \label{eq:sop_cs_beta}\\
	&p^*_j \left( \sum_{i=1}^n x^*_{i,j} - 1\right) = 0,&\forall j \in [m], \label{eq:sop_cs_p}\\
	&x^*_{i,j}(p^*_j -w^*_i v_{i,j}) = 0,&\forall i \in [n],j \in [m], \label{eq:sop_cs_x}\\
	&d^*_i\left(w^*_i - \frac{1}{\tau_i} \right) = 0,&\forall i \in [n]. \label{eq:sop_cs_t}
\end{align}

We first establish key structural properties of the optimal solutions. The following result provides a lower bound for the optimal $w^*_i$.
% Before we dive into our analysis, we first prove some useful properties of the optimal solutions.
\begin{lemma}[Lower Bound on Dual Variables]\label{lemma:beta_lower_bound}
    Define $\underline{w}_i = \min \left\{\frac{\lambda_i}{m \bar{v}}, \frac{1}{\tau_i} \right\}, i \in [n].$
    For any optimal solution $(w^*,p^*)$ to Program \eqref{prog: dual}, $w^*_i$ is lower bounded by $\underline{w}_i$, i.e., $w^*_i \ge \underline{w}_i,\forall i \in [n]$.
    % According to the program (\ref{prog: primal}) and the program (\ref{prog: dual}), we have for each buyer, $i \in [n]$, $w_i \ge \underline{w}_i$
\end{lemma}

Any optimal solution must satisfy the feasibility constraint $p^*_j\ge w^*_iv_{i,j},\forall i\in[n],j\in[m]$ in Program \eqref{prog: dual}. In fact, we must have $p^*_j=\max_i \{w^*_iv_{i,j}\}$, as lowering $p^*_j$ clearly improves the objective of the program. And with Lemma \ref{lemma:beta_lower_bound} and the assumption that $\max_i\{v_{i,j}\}>0$, we get $p_j^* > 0,\forall j\in[m]$, i.e., the prices are all positive. Combined with Equation \eqref{eq:sop_cs_p}, we have that $\sum_{i=1}^n x_{i,j}^* = 1$, i.e., the market fully clears.

We then show that the solution to Program \eqref{prog: primal} ensures that each individual's budget and RoS constraints are satisfied.
\begin{lemma}[Feasibility of Constraints]\label{lemma:feasible}
    Let $x^*_{i,\cdot}$ be the allocation of buyer $i$, and $t^*_i=\sum_{j=1}^mp^*_jx^*_{i,j}$ the corresponding payment. Then $x^*_{i,\cdot}$ and $t^*_i$ satisfy each buyer's budget and RoS constraints.
\end{lemma}

The allocation structure is determined by the complementary slackness condition in Equation (\eqref{eq:sop_cs_x}) together with Definition \ref{def:modified_competitive_equilibrium}. According to Equation (\eqref{eq:sop_cs_x}), buyer $i$ can only receive a positive allocation of item $j$ (i.e., $x_{i,j}^* > 0)$ if their weighted valuation equals the market price: $w_i^* v_{i,j} = p_j^*$.

The realized allocation takes one of two forms depending on whether ties occur:
\begin{itemize}
    \item \textbf{Generic Case (Unique Buyer)}: When a unique buyer $i$ strictly maximizes the weighted valuation for item $j$, the allocation is deterministic and binary:
    \begin{align*}
        x_{i,j}^* = \begin{cases}
            1 & \text{if } w_i^*v_{i,j} = p_j^* \text{ and } w_k^*v_{k,j} < p_j^*, \forall k \neq i,\\
            0 & \text{otherwise.}
        \end{cases}
    \end{align*}
    In this case, the seller sets $y_{i,j}^* = 1$ and $y_{k,j}^* = 0$ for all $k \neq i$.
    \item \textbf{Tie Case (Multiple Maximizers)}: When multiple buyers $\mathcal{S} = \{i \in [n]: w_i^*v_{i,j} = p_j^* \}$ have identical maximal weighted valuations for item $j$, the seller acts as a coordinator by strategically choosing supplies $\{y_{i,j}^*\}_{i \in \mathcal{S}}$ satisfying $\sum_{i\in \mathcal{S}} y_{i,j}^* = 1$ to ensure market clearance across all items. By Equation (\eqref{eq:sop_cs_x}), each buyer $i \in \mathcal{S}$ demands exactly the supply allocated to him $x_{i,j}^* = y_{i,j}^*$. This coordination mechanism guarantees that supply equals demand even when RoS constraints induce discontinuous demand functions.
\end{itemize}

The prices $p^*$ derived from the extended Eisenberg-Gale program may not constitute a traditional competitive equilibrium without the seller's coordination role. In classical markets with only budget constraints, if an item's price is too low, excess demand is resolved automatically as the price rises and buyers' demands decrease continuously until supply equals demand. However, with RoS constraints, buyers' demand functions exhibit discontinuities—a buyer's demand can drop abruptly from 1 to 0 when the price crosses their RoS threshold. This creates situations where at certain prices, multiple buyers simultaneously demand the same item, yet no price adjustment can resolve the tie without eliminating all demand.

This is precisely where Definition \ref{def:modified_competitive_equilibrium} becomes essential. By allowing the seller to strategically allocate fractional supplies when ties occur, we ensure that a market-clearing equilibrium always exists. The seller leverages its role as a centralized coordinator to balance supply and demand across the entire market, resolving conflicts that would otherwise make equilibrium non-existent.

% Then we show that the price $p^*$, along with the allocation $x^*$, form a competitive equilibrium.
Given a set of prices $\{p_j\}_{j\in [m]}$, each buyer's buying decision should aim to maximize their utility, subject to both the budget and RoS constraints. Specifically, buyer $i$ solves the individual optimization program (\ref{IOP:primal}).
Given prices $\{p_j\}_{j\in [m]}$ and seller's supplies $\{y_{i,j}\}_{i\in[n], j\in[m]}$, each buyer's buying decision aims to maximize their utility subject to both the budget and RoS constraints. Specifically, buyer $i$ solves the individual optimization program (\ref{IOP:primal}), where the constraint $x_{i,j}\in [0, 1+(y_{i,j}-1)\cdot \mathbb{I}_{0<y_{i,j}<1}]$ reflects the buyer's demand flexibility: if the seller offers a fractional supply ($0<y_{i,j}<1$), the buyer can demand up to that amount; otherwise, the buyer is constrained to $[0,1]$.

We now show that if prices are set to $p=p^*$ and supplies are set to $y_{i,j} = x_{i,j}^*$ for all $i\in[n], j\in[m]$, then $x^*_{i,\cdot}$ is indeed buyer $i$'s optimal demand.

\begin{lemma}\label{lemma:SOP_IOP}
    If the prices are $p=p^*$ and the seller's supply is $y_{i,j} = x_{i,j}^*$ for all $i\in[n], j\in[m]$, then $x^*_{i,\cdot}$ is an optimal solution to Program \eqref{IOP:primal} for buyer $i$.
\end{lemma}

\begin{remark}[Role of Definition \ref{def:modified_competitive_equilibrium}]
Definition \ref{def:modified_competitive_equilibrium} is essential for ensuring the existence of competitive equilibrium in our setting. Unlike markets with only budget constraints—where fractional allocations naturally emerge through price adjustments to balance supply and demand—RoS constraints introduce fundamental discontinuities in buyer demand functions.

The issue arises because a buyer's willingness to purchase depends on the value-to-price ratio meeting their RoS threshold. When the price of an item is exactly at the RoS threshold for multiple buyers (i.e., $\frac{v_{i,j}}{p_j} = \tau_i$ for multiple $i$), each such buyer either demands the full unit or nothing—their demand jumps discontinuously. At this critical price point:
\begin{itemize}
    \item Raising the price even slightly causes all RoS-constrained buyers to drop their demand to zero (violating their RoS constraint);
    \item Lowering the price causes excess demand (more than one unit total);
    \item No price adjustment can resolve the tie through the traditional equilibrium mechanism.
\end{itemize}

Definition \ref{def:modified_competitive_equilibrium} resolves this impasse by empowering the seller to act as a market coordinator. When ties occur, the seller can strategically allocate fractional supplies $\{y_{i,j}\}$ to the tied buyers, ensuring global market clearance. Crucially, at the tie price, buyers are value-maximizers, so accepting the seller's allocated supply is optimal for them. Example \ref{example:tie_breaking_rule} illustrates this mechanism in detail.
\end{remark}

Lemma \ref{lemma:SOP_IOP} shows that if the prices are set according to $p^*$, combining with the seller's coordinated supply $y^*$, each buyer's decision is the same as $x^*_{i,\cdot}$. Now we prove $p^*, y^*$ and $x^*$ together form a competitive equilibrium.
\begin{lemma}
    \label{lemma:competitive equilibrium}
    The pair $(p^*,x^*, y^*)$ forms a competitive equilibrium of the auction market.
\end{lemma}

According to Lemma \ref{lemma:competitive equilibrium}, we can find a competitive equilibrium of the auction market by solving Program \eqref{prog: primal} and \eqref{prog: dual}. The following result shows that any competitive equilibrium also corresponds to an optimal solution to Program \eqref{prog: primal} and \eqref{prog: dual}.

% Then we show that if there exists a price $p$ and each buyer selects his favorite bundle via IOP. If for each item, all buyers' demands cannot exceed the seller's supply, then this price and each buyer's selection is the solution of the CAP.

\begin{lemma}\label{lemma:IOP_SOP}
	% If there exists $(p, x)$ such that each buyer maximizes his own utility then $x$ is also the optimal solution of the program (\ref{prog: primal}), and $p$ is the optimal solution of the program (\ref{prog: dual})
    Let $(p^*, x^*, y^*)$ be any competitive equilibrium of the auction market. Then there exists $\{d^*_i\}_{i\in [n]}$, $\{c^*_i\}_{i\in [n]}$, and $\{w^*_i\}_{i\in [n]}$, such that $(x^*,d^*,c^*)$ and $(p^*,w^*)$ are optimal solutions to Program \eqref{prog: primal} and \eqref{prog: dual}, respectively.
\end{lemma}

\begin{remark}[Role of Assumption \ref{assumption:no_buyer_win_all}]
    Assumption \ref{assumption:no_buyer_win_all} is important for the proof of Lemma \ref{lemma:IOP_SOP}. Without this assumption, the environment becomes non-competitive, and it is possible for a buyer $i$ to buy all items without violating their budget and RoS constraints when $\tau_i=1$ and $\lambda_i$ is sufficiently large.  %In this case, any price that only this specific buyer can buy but other buyers cannot afford forms a competitive equilibrium, which implies that there are multiple competitive equilibriums. Under Assumption \ref{assumption:no_buyer_win_all}, the competitive equilibrium is unique.
\end{remark}

Having characterized the competitive equilibrium, we now establish its uniqueness and Pareto efficiency of the optimal solution to Program \eqref{prog: primal} and \eqref{prog: dual}. We first show that optimal solution to Program \eqref{prog: primal} is unique. Before that, we provide the following technical lemma, which is also useful in Section \ref{sec:computation}.
% After showing the lower bound of each multiplier, we further show that the objective function of the program (\ref{prog: dual}) is strongly convex.

% 改下叙述
\begin{lemma}[Strong Convexity of the Regularizer]\label{lemma:strong_convex}
Let the regularizer function be defined as:
\begin{align*}
h(w) = - \sum_{i=1}^n \lambda_i \log (w_i),
\end{align*}
over the domain $W = \{w \in \mathbb{R}^n \mid 0 < w_i \le \frac{1}{\tau_i}, \forall i \in [n]\}$.
The function $h(w)$ is $\sigma$-strongly convex with respect to the $\ell_2$-norm. Specifically, for any $u, v \in W$:
\begin{equation}
h(v) \ge h(u) + \langle \nabla h(u), v - u \rangle + \frac{\sigma}{2} \|v - u\|_2^2,
\end{equation}
where $\sigma = \min_{i\in [n]} \lambda_i \tau_i^2$, $\langle \cdot,\cdot\rangle$ denote the inner product of two vectors, and $\|\cdot\|_2$ denote the $\ell_2$-norm.
\end{lemma}

The uniqueness of the optimal solution to Program \eqref{prog: primal} is almost a direct implication of Lemma \ref{lemma:strong_convex}.
\begin{lemma}[Revenue Uniqueness]\label{lemma:unique}
    While there may exist multiple allocation solutions $(x^*, y^*)$ to the competitive equilibrium (corresponding to different ways the seller can coordinate fractional supplies when ties occur), all such equilibria generate the same revenue. That is, for any two competitive equilibria $(x^*, p^*, y^*)$ and $(\tilde{x}, \tilde{p}, \tilde{y})$, we have $p^* = \tilde{p}$ and $\sum_{j=1}^m p_j^* = \sum_{j=1}^m \tilde{p}_j$.
\end{lemma}

Similar to the standard Fisher market model, the optimal solution to Program \eqref{prog: primal} also satisfies the Pareto efficiency property.

% After showing the existence, and uniqueness of the auto-bidding equilibrium in first-price auctions, we further show some properties, which the equilibrium satisfies.
\begin{lemma}\label{lemma:Pareto_efficiency}
    Let $(x^*, d^*, c^*)$ be the optimal solution to Program \eqref{prog: primal}. Then $x^*$ is Pareto efficient, i.e., no other feasible allocation can weakly improve all buyers' utilities and strictly increase at least one buyer's utility.
\end{lemma}

We now present our first main result.
\begin{theorem}[Equilibrium Characterization]
    \label{thm:competitive equilibrium characterization}
    Let $(x,d,c)$ and $(p,w)$ be feasible solutions to Program \eqref{prog: primal} and \eqref{prog: dual}, respectively. They are both optimal solutions if and only if $(x,p,y)$ forms a competitive equilibrium, where $y = x$. Moreover, any competitive equilibrium is Pareto efficient, and all equilibria yield the same unique revenue and prices.
\end{theorem}

Theorem \ref{thm:competitive equilibrium characterization} can be obtained directly by combining Lemma \ref{lemma:competitive equilibrium}, Lemma \ref{lemma:IOP_SOP}, Lemma \ref{lemma:unique} and Lemma \ref{lemma:Pareto_efficiency}.

Building on the characterization above, we construct a mechanism that implements the competitive equilibrium outcome. The mechanism is formally defined as follows.

\begin{definition}[Market-Clearing Mechanism]\label{def:market_clearing_mechanism}
For buyers reporting their constraints $(\lambda, \tau)$ truthfully and public valuations $v$, let $(x^*, p^*)$ denote one competitive equilibrium computed via the extended Eisenberg-Gale program.
\begin{itemize}
    \item \textbf{Allocation Rule:} The mechanism allocates items according to the equilibrium demand.
    \begin{equation*}
        a_{i,j}(v_{\cdot, j}, \lambda,\tau) = x_{i,j}^*, \quad \forall  i \in [n], j\in[m].
    \end{equation*}
    \item \textbf{Payment Rule}: The mechanism charges the market-clearing price for allocated items:
    \begin{equation*}
        t_{i,j}(v_{\cdot, j}, \lambda, \tau) = p_j^*\cdot x_{i,j}^*, \quad \forall i\in [n], j\in[m].
    \end{equation*}
\end{itemize}
\end{definition}

Under this mechanism, buyers effectively face a set of personalized, market-clearing prices that balance global supply and demand.

Recall that the complementary slackness condition Equation \eqref{eq:sop_cs_x} implies that in the optimal solution, any buyer allocated with item $j$ (i.e., $x^*_{i,j}\ne 0$) must satisfy $w^*_iv_{i,j}=p^*_j$. The first constraint in Program \eqref{prog: dual} requires $p^*_j\ge w^*_iv_{i,j},\forall i\in[n],j\in[m]$. Consequently, we have
% According to the dual program (\ref{prog: dual}), we know that 
\begin{align}\label{eq:price_characterization}
	p^*_j = \max_{i\in[n]} w^*_i v_{i,j}, \forall j \in [m].
\end{align}

The above equation indicates that if buyer $i$ uses $w^*_i$ as their strategy ($\omega_i=w^*_i$ and $b_{i,j}=w^*_iv_{i,j}$), then the price of item $j$ is the highest bid among all buyers. This characterization reveals a direct equivalence to the first-price auctions. If we interpret the optimal dual variable $w_i^*$ as buyer $i$'s multiplicative bidding strategy, where buyer $i$ submits a bid $b_{i,j} = w_i^* v_{i,j}$, then Equation (\ref{eq:price_characterization}) states that the item is sold to the highest bidder at a price equal to their bid. 

Thus, the abstract market-clearing mechanism can be implemented via first price auctions with an uniform bidding strategy, where each buyer's strategy is a multiplier across all auctions. The market-clearing mechanism's outcome is the same as the equilibrium outcome of first price auctions with a uniform bidding strategy. We formalize this implementation as the following game:
\begin{definition}[Auto-Bidding Implementation]
    The auto-bidding game involves $n$ buyers and a central auto-bidding mechanism. The game proceeds as follows:
    \begin{itemize}
    \item Each buyer submits their budget and target RoS $(\lambda_i, \tau_i)$ truthfully to the seller; 
    \item The seller solves Program \eqref{prog: dual} to determine their multipliers $w^* = \left\{w_i^*\right\}_{i\in[n]}$ and use $w_i^* v_{i,j}, \forall i\in[n], j\in[m]$ as bids to participate in the auctions and the item is allocated according to standard first-price auction rules.
\end{itemize}
\end{definition}

We establish that the mechanism is incentive compatible. \citet{alimohammadi2023incentive} proved a related result for first-price auctions with a uniform-bidding strategy. Our Theorem \ref{thm:ic} provides an alternative proof from the market equilibrium perspective.

\begin{theorem}
    \label{thm:ic}
    Let $x^*, p^*$ be any optimal solutions to Program \eqref{prog: primal} and Program \eqref{prog: dual}, respectively. Then each buyer reporting $(\lambda_i, \tau_i)$ truthfully is an optimal strategy when other buyers report their financial constraints $(\lambda_{-i}, \tau_{-i})$ truthfully. In other words, the market-clearing mechanism defined in Definition \ref{def:market_clearing_mechanism} is incentive compatible.
\end{theorem}

\subsection{Revenue Guarantee of the Market-Clearing Mechanism}
We now analyze the revenue of the market-clearing mechanism. We begin by defining our benchmark: the first-best revenue.
\begin{definition}[First-Best Revenue]
    Let $(x^{FB}, t^{FB})$ be the optimal solution to the following first-best program: 
    \begin{maxi}|l|[2]{x,t}{\sum_{i=1}^n t_i}{\tag{FB-PRIMAL}\label{prog:fb_primal}}{}
    \addConstraint{t_i}{\le \lambda_i}{\quad\forall i \in [n]}
    \addConstraint{t_i}{\le \frac{1}{\tau_i} \sum_{j=1}^m v_{i,j} x_{i,j}}{\quad\forall j \in [m]}
    \addConstraint{\sum_{i=1}^n x_{i,j}}{\le 1}{\quad \forall j \in [m]}
    \addConstraint{x_{i,j}}{\ge 0}{\quad\forall i \in [n], \forall j \in [m]}.
    \end{maxi}
    The first-best revenue is defined as the optimal objective, i.e., $Rev^{FB}=\sum_{i=1}^n t^{FB}.$
    We also call $x^{FB}$ the first-best allocation.
\end{definition}

Note that the mechanism induced by the first-best program may not be incentive compatible, as no IC constraints are included in the program. Instead, its feasible region is simply all the allocations that respect each buyer's budget and RoS constraints. Therefore, the first-best revenue is the maximum possible revenue achievable by a feasible allocation, and can serve as an upper bound of the revenue of any IC mechanism.

Before presenting our revenue approximation result, we first introduce some notations that will be useful for later arguments. Given any outcome $(x,t)$, define
\begin{gather*}
    \mathcal{B}(x,t)=\left\{ i\,\middle|\, \sum_{j=1}^m t_{i,j}=\lambda_i, \sum_{j=1}^mv_{i,j}x_{i,j}\ge \tau_i \sum_{j=1}^mt_{i,j}x_{i,j} \right\},\\
    \mathcal{R}(x,t)=\left\{ i\,\middle|\, \sum_{j=1}^m t_{i,j}< \lambda_i, \sum_{j=1}^mv_{i,j}x_{i,j}= \tau_i \sum_{j=1}^mt_{i,j}x_{i,j}\right\}
\end{gather*}
to be the set of budget-constrained buyers and the set of RoS-constrained buyers, respectively. For simplicity, we denote by $\mathcal{B}^*,\mathcal{R}^*$ and $\mathcal{B}^{FB},\mathcal{R}^{FB}$ the corresponding sets under the market-clearing mechanism $(x^*,t^*)$ and the first-best outcome $(x^{FB}, t^{FB})$. It is straightforward to check that under both $(x^*,t^*)$ and $(x^{FB}, t^{FB})$, each buyer is either a budget-constrained buyer or an RoS-constrained buyer, since otherwise, there exists a buyer with both non-binding budget constraints and RoS constraints, and we can always charge the buyer more to achieve a larger objective without violating their financial constraints. The revenue guarantee of the market-clearing mechanism relies on a careful analysis of Program \eqref{prog:fb_primal} and it dual program, which can be written as follows:
\begin{mini}|l|[2]{\alpha, \beta, p}{\sum_{i=1}^n \alpha_i \lambda_i + \sum_{j=1}^m p_j}{\tag{FB-DUAL}\label{prog:fb_dual}}{}
\addConstraint{\alpha_i + \tau_i \beta_i}{\ge 1}{\quad \forall i \in [n]}
\addConstraint{p_j}{\ge \beta_i v_{i, j}}{\quad \forall i \in [n], j \in [m]} 
\addConstraint{\alpha_i}{\ge 0}{\quad \forall i \in [n]}
\addConstraint{\beta_i}{\ge 0}{\quad \forall i \in [n]}
\end{mini}

Let $(\alpha^{FB}, \beta^{FB},p^{FB})$ be an optimal solution to the above dual program. The second constraint of Program \eqref{prog:fb_dual} shows that the first-best allocation can also be viewed as the result of a ``first-price auction'' and item $j$ is allocated to the buyer with the highest ``bid'' $\max_i\{\beta^{FB}_iv_{i,j}\}$, where $\beta^{FB}_i$ is a uniform coefficient across all items. However, this highest ``bid'' cannot be regarded as the price of the item, and the payment of a buyer is also not the sum of their bids on winning items. 
% Our revenue approximation result depends on the following technical lemma.
% \begin{lemma}\label{lemma:relationship}
%     For any buyer $i$, there exists an optimal solution $(\alpha^{FB}, \beta^{FB},p^{FB})$ to Program \eqref{prog:fb_dual}, where the bidding coefficient $\beta^{FB}$ is weakly smaller than that in the market clearing mechanism, i.e.,
%     \begin{gather*}
%         \beta_i^{FB} \le w_i^*, \forall i.
%     \end{gather*}
%     Specifically, for bidders in $\mathcal{R}^{FB}$, the equality holds, i.e.,
%     \begin{gather*}
%         \beta_i^{FB} = w_i^* = \frac{1}{\tau_i}, \forall i \in \mathcal{B}^{FB}.
%     \end{gather*}
% \end{lemma}

Now we present our revenue approximation result.
\begin{theorem}[Revenue Guarantee of the Market-Clearing Mechanism]\label{thm:revenue_approximation}
    Let $\text{Rev}^*$ and $\overline{\text{Rev}}$ denote the revenue of the market-clearing mechanism and the revenue-optimal mechanism. The revenue of the market-clearing mechanism is a $\frac{1}{2}$-approximation of that of the first-best outcome, i.e., $\text{Rev}^{*} \ge \frac{1}{2} \text{Rev}^{FB}$.
    An immediate consequence is that the revenue of the market-clearing mechanism is also a $\frac{1}{2}$-approximation of the optimal revenue $\text{Rev}^{*} \ge \frac{1}{2} \overline{\text{Rev}}$.
\end{theorem}
In the special case where all buyers budget-constrained buyers, the mechanism extracts the full budget from every buyer, which clearly maximizes the seller's revenue. Conversely, when all buyers are RoS-constrained buyers, the budget constraints become useless and the mechanism's outcome is the same as the first-best revenue, which is also known to be revenue-optimal when all buyers have only RoS constraints \cite{balseiro2021landscape}. In the general case, our result shows that the seller's revenue is at least $\frac{1}{2}$ of the optimal revenue. Furthermore, we provide a concrete instance (see Example \ref{ex:revenue_tight}) demonstrating that this approximation bound is tight.

% The idea behind the result is to use $\text{Rev}^{FB}$ as an intermediary. Since $\text{Rev}^{FB}$ is an upper bound of any IC mechanism, it suffices to show that $\text{Rev}^{*}\ge \frac{1}{2}\text{Rev}^{FB}$.

The intuition for the approximation ratio stems from comparing the allocation rule implied by the KKT conditions of the extended Eisenberg-Gale program and the first-best revenue program. Both the market-clearing mechanism and the first-best (FB) program allocate items according to a ``weighted valuation'' rule:
\begin{itemize}
    \item The market-clearing mechanism: Allocating item $j$ to the buyer with the maximized $w_i^* v_{i,j}, \forall i\in [n]$, where $w^* = (w_i^*)_{i\in[n]}$ are optimal solutions of Program \eqref{prog: dual}.
    \item The first-best program: Allocating item $j$ to the buyer with maximized $\beta_i^{FB}v_{i,j}, \forall i\in[n]$, where $\beta_i^{FB}$ is the optimal dual variable of the first-best revenue program.
\end{itemize}

Crucially, for RoS-constrained buyers, the weights are identical in both settings $(w_i^* = \beta_i^{FB} = \frac{1}{\tau_i})$. However, for budget-constrained buyers, $w_i^*$ is typically larger than the first-best weights $\beta_i^{FB}$. This high weight allows budget-constrained buyers to obtain more items while not paying more.

The revenue gap arises because the market-clearing mechanism collects only the budget from budget-constrained buyers with more items, while under the first-best revenue, some of their items might be reallocated to RoS-constrained buyers for higher revenue. However, because the budget-constrained buyers won these items in the market-clearing mechanism, their weighted valuations were higher than the weighted valuations of the RoS-constrained buyers. This dominance implies that the additional revenue from RoS-constrained buyers in the first-best revenue is upper-bounded by the weighted value of the budget-constrained buyers under the market-clearing mechanism, which corresponds exactly to the total budgets of those budget-constrained buyers. Thus, the total first-best revenue cannot exceed the sum of the extended EG program's revenue plus the budget-constrained buyers' total budgets, yielding a factor of $\frac{1}{2}$.

We provide an example in the appendix (Example \ref{ex:revenue_tight}) to show that the bound is tight.

\section{Online Implementation of the Market-Clearing Mechanism}\label{sec:computation}

The competitive equilibrium characterized in Section \ref{sec:theory} assumes the seller can solve the global convex Programs \eqref{prog: primal} and \eqref{prog: dual} with complete knowledge of all $m$ auctions. However, real-world advertising platforms face a fundamentally different challenge: auctions arrive sequentially, and future valuations are unknown at decision time. 

To address this challenge, we propose an online learning algorithm that computes the optimal bidding multipliers on the fly. The key insight is that the dual program \eqref{prog: dual} has a separable structure: the objective decomposes into per-auction revenue plus a strongly-convex regularizer (Lemma \ref{lemma:strong_convex}). This structure enables us to apply Regularized Dual Averaging (RDA)~\cite{xiao2010dual}, an online convex optimization technique, to incrementally learn $w^*$ using only historical data. We demonstrate that this approach achieves sub-linear regret relative to the optimal static solution. Furthermore, the algorithm is fully decentralized: each buyer can independently optimize their bidding strategy using only their own auction history, yet the collective system provably converges to the market-clearing equilibrium.

\subsection{Online Problem Setup}

We consider a sequence of $m$ auctions. To facilitate asymptotic analysis, we assume the total budget of each buyer $i$ scales linearly with the auction size $m$. This assumption is commonly used in the online learning literature~\cite{agrawal2014dynamic,li2022online,gao2021online}. Specifically, let $ m\rho_i = \lambda_i$, where $\rho_i \ge \underline{\rho} > 0$ represents the average per-auction budget. The lower bound $\underline{\rho}$ ensures that budgets are not vanishingly small relative to the number of items. Recall that the dual program \eqref{prog: dual} minimizes the aggregate price minus a logarithmic barrier term. In the online setting, we can express the global objective function as minimizing the cumulative revenue plus a regularizer:

\begin{gather}
    \mathcal{L}(\omega) = \sum_{j=1}^m \hat{p}_j(\omega) - \sum_{j=1}^m \sum_{i=1}^n \rho_i \log(\omega_i),
\end{gather}
where $\hat{p}_j(\omega) = \max_{i \in [n]} \{ \omega_i v_{i,j} \}$ represents the seller's revenue in auction $j$ given bidding multipliers $\omega$. The term $-\rho_i \log(\omega_i)$ acts as a strictly convex regularizer that penalizes multipliers approaching zero, which would correspond to an infinite objective value.

\subsection{The Online Algorithm}

The intuition behind our approach is to apply the RDA framework to minimize $\mathcal{L}(\omega)$. In each round $j$, the algorithm maintains a dual average vector $\bar{g}^j$, which aggregates the function $\hat{p}_j(\omega)$ observed so far.

Let $g^j(\omega^j) \in \mathbb{R}^n$ denote the sub-gradient of function $\hat{p}_j(\cdot)$ at auction $j$. For each buyer $i$, this component is simply the value of the winning item (or zero if lost):\begin{gather*}g_{i}^j(\omega^j) = v_{i,j} x_{i, j}(\omega^j),\end{gather*}where $x_{i,j}(\omega^j)$ indicates whether buyer $i$ won item $j$. We define the cumulative dual average $\bar{g}_i^j$ as:

\begin{gather}\label{eq:dual_average}
\bar{g}_i^j = \frac{j-1}{j} \bar{g}_i^{j-1} + \frac{1}{j} g_i^j(\omega^j).
\end{gather}

Economically, $\bar{g}_i^j$ represents buyer $i$'s average realized value per auction up to auction $j$. Based on this history, the algorithm updates the bidding strategy $\omega_i$ to minimize the linearized objective function combined with the regularizer. The specific update rule and the complete procedure are detailed in Algorithm \ref{alg:RDA}. We define the feasible strategy interval for buyer $i$ as $W_i = [\underline{w}_i, 1/\tau_i]$, where $\underline{w}_i = \min \left\{ \frac{\rho_i}{\bar{v}}, \frac{1}{\tau_i}\right\}$ is a lower bound of $w_i^*$, as shown in Lemma \ref{lemma:beta_lower_bound}.

\begin{algorithm}[!ht]
\SetAlgoNoLine\KwIn{Per-auction budget $\rho_i$, target RoS $\tau_i$, and feasible set $W_i$ for each buyer $i \in [n]$.}\KwOut{Allocation $x$ and strategy sequence $\omega$.}\textbf{Initialize:} $\omega_i^1 \gets \frac{1}{\tau_i}, \quad \bar{g}_i^0 \gets 0, \quad \forall i \in [n]$\;
\For{each auction $j = 1, \dots, m$}{
    \textbf{1. Allocation (Auction Execution):}\\
    The platform computes the allocation $x_{\cdot, j}$ based on current bids $b_{i,j} = \omega_i^j v_{i,j}$:
    \begin{algomathdisplay}
        x_{i,j} = \begin{cases}
            1 & \text{if } i \in \argmax_{k \in [n]} \{ \omega_k^j v_{k, j} \} \\
            0 & \text{otherwise}
        \end{cases}
    \end{algomathdisplay}
    
    \textbf{2. Update Dual Average (History Aggregation):}\\
    Each buyer updates their average realized value:
    \begin{algomathdisplay}
        \bar{g}_i^j = \frac{j-1}{j} \bar{g}_i^{j-1} + \frac{1}{j} (v_{i,j} x_{i,j})
    \end{algomathdisplay}
    
    \textbf{3. Strategy Update (Primal-Dual Step):}\\
    Each buyer computes the next multiplier $\omega_i^{j+1}$ by solving:
    \begin{algomathdisplay}
        \omega_i^{j+1} = \argmin_{\omega \in W_i} \left\{ \bar{g}_i^j \cdot \omega - \rho_i \log(\omega) \right\}
    \end{algomathdisplay}
    \text{Closed-form solution (Projected Newton Step):}
    \begin{algomathdisplay}
        \omega_i^{j+1} = \text{Proj}_{\left[\underline{w}_i, \frac{1}{\tau_i}\right]} \left( \frac{\rho_i}{\bar{g}_i^j} \right) = \begin{cases}
            \underline{w}_i & \text{if } \frac{\rho_i}{\bar{g}_i^j} \le \underline{w}_i \\
            \frac{1}{\tau_i} & \text{if } \frac{\rho_i}{\bar{g}_i^j} \ge \frac{1}{\tau_i} \\
            \frac{\rho_i}{\bar{g}_i^j} & \text{otherwise}
        \end{cases}
    \end{algomathdisplay}
}
\caption{Decentralized Online Market-Clearing Algorithm}
\label{alg:RDA}
\end{algorithm}

A key feature of Algorithm \ref{alg:RDA} is its decomposability. The update rule for buyer $i$ (Step 3) depends exclusively on their own parameters $(\rho_i, \tau_i)$ and their own private auction history ($\bar{g}_i^j$). It does not require knowledge of other buyers' budgets, valuations, or current multipliers. This property implies that the market-clearing mechanism can be implemented in a fully decentralized manner: the platform simply runs a standard First-Price Auction, and each buyer (or their auto-bidder) independently runs Algorithm \ref{alg:RDA} to adjust their bids. Thus, Algorithm \ref{alg:RDA} preserves buyers' privacy. Buyers do not need to reveal their budgets, RoS constraints, or bidding strategies to others.

Despite this complete decentralization, Theorem \ref{thm:alg} shows that the collective behavior of all buyers running Algorithm \ref{alg:RDA} independently converges to the centralized competitive equilibrium at sub-linear rate.

% In the following analysis, we first show that this algorithm achieves $O(\log(m))$ regret for the objective function in the program (\ref{prog: dual_online})
\subsection{Convergence Analysis of Online Market-Clearing Mechanism}

To analyze the convergence properties of Algorithm \ref{alg:RDA}, we require an additional assumption about the valuation distribution. 
\begin{assumption}[Non-Atomic Valuations for First-Price Auction Implementation]\label{assumption:continuous_valuation}
    For the first-price auction implementation, we assume that for any item $j \in [m]$, the valuation $v_{i,j}$ of each buyer $i$ is drawn independently from a continuous (non-atomic) distribution.
\end{assumption}

\begin{remark}[Role of Assumption \ref{assumption:continuous_valuation}]
This assumption serves a different purpose than Definition \ref{def:modified_competitive_equilibrium}.
Assumption \ref{assumption:continuous_valuation} addresses ties in the \emph{practical first-price auction implementation}. Under this assumption, for any two distinct buyers $i, k \in [n]$, the probability that their weighted bids are identical (i.e., $w_i^*v_{i,j} = w_k^*v_{k,j}$) is zero. Consequently, ties in first-price auctions occur with probability zero, eliminating the need for explicit tie-breaking rules in the implementation.

Importantly, Assumption \ref{assumption:continuous_valuation} is introduced solely to simplify the practical auction implementation by precluding ties.
\end{remark}

This assumption is standard in the auction literature~\citep{myerson1981optimal,balseiro2021landscape,balseiro2022optimal}. In the context of online advertising, valuations are typically products of cost per action, predicted click-through rates (pCTR), and conversion rates (pCVR). Since these underlying probabilities vary continuously across heterogeneous users and ad creatives, the probability that two distinct buyers have identical valuations for one item is zero.

With this assumption in place, we now proceed to analyze the convergence guarantees of Algorithm \ref{alg:RDA}.
Let $w^* = (w_1^*, \dots, w_n^*)$ denote the optimal strategy profile derived from the offline dual program \eqref{prog: dual}. Since the platform cannot compute $w^*$ in advance, our goal is to design an online algorithm that approaches this benchmark. To quantify the performance of our algorithm, we define the auxiliary regret as the cumulative difference between the objective value achieved by the online strategy sequence and that of the fixed optimal strategy $w^*$. This metric serves as a key intermediate step for bounding both the seller's revenue regret and the buyers' utility regret.

\begin{definition}[Auxiliary Regret]\label{def:auxiliary_regret}
    The auxiliary regret up to the $j$-th auction is defined as follows:
    \begin{align*}
    R^j_{obj} = \sum_{k=1}^j \left[\hat{p}_{k}(\omega^{k}) - \sum_{i=1}^n \rho_i \log(\omega_i^{k}) \right] - \sum_{k=1}^j \left[\hat{p}_{k}(w^*) - \sum_{i=1}^n \rho_i \log(w^*_i) \right],
\end{align*}
where $\omega^k = (\omega_1^k, \cdots, \omega_n^k)$ denotes the strategy vector learned by the algorithm for the $k$-th auction and $\hat{p}_k(\omega) = \max_{i \in[n]} \left\{\omega_i v_{i,k} \right\}$ as the highest bid in auction $k$ when buyers use strategy $\omega$.
\end{definition}
% We also define each buyer's regret as follows.
% \begin{definition}[Buyer Regret]
%     The regret of buyer $i$ up to the $j$-th auction is defined as follows:
%     \begin{align*}
%         R^j_{u} = \left|u_i^j - u_i^* \right|=\left|\sum_{j' = 1}^j v_{i,j'} x_{i, j'} - \sum_{j' = 1}^j v_{i, j'} x_{i, j'}^*\right|,
%     \end{align*}
%     where $u_i^j$ and $u_i^*$ denote buyer $i$'s cumulative utility under allocation $x$ and $x^*$ up to the $j$-th auction.
% \end{definition}

Using the same technical methods in \citet{xiao2010dual}, we first prove the bound of the auxiliary regret.
\begin{theorem}[Auxiliary Regret Bound]\label{thm:auxiliary_regret}
    The auxiliary regret up to the $m$-th auction is bounded by $O\left(\log(m)\right)$, i.e., 
    \begin{equation*}
        R_{obj}^m \le O \left(\log(m) \right).
    \end{equation*}
\end{theorem}

After bounding the auxiliary regret, we further bound each buyer's strategy $\omega_i^{m+1}$ after $m$ auctions.

\begin{lemma}\label{lemma:strategy_comparison}
    For each buyer $i$, comparing to the optimal uniform-bidding strategy $\omega_i^*$, the output of Algorithm \ref{alg:RDA} satisfies
    \begin{align*}
        \left|\omega^{m+1}_i - w_i^*\right| \le O\left(\sqrt{\frac{\log(m)}{m}} \right). 
    \end{align*}
\end{lemma}

\begin{lemma}[Each Buyer's Utility Bound]\label{lemma:utility_bound}
    If each buyer's value is drawn from a continuous distribution as we mentioned in assumption \ref{assumption:continuous_valuation}, and each buyer's density distribution is bounded by $\bar{f}$, i.e., $f_i(\cdot) \le \bar{f}, \forall i \in[n]$, then, for any buyer $i \in [n]$, algorithm \ref{alg:RDA} ensures buyer $i$'s utility as
    \begin{align*}
       u_i^* - u_i^m \le O\left(\sqrt{m\log(m)}\right),
    \end{align*}
    where $u_i^m = \sum_{j=1}^m v_{i,j}x_{i,j}$ denotes buyer $i$'s cumulative utility with algorithm \ref{alg:RDA}  and $u_i^* = \sum_{j=1}^m v_{i,j}x_{i,j}^*$ denotes buyer $i$'s optimal utility.
\end{lemma}

Further, we bound the seller's revenue.
\begin{lemma}[Seller's Revenue Bound]\label{lemma:seller_revenue}
    Algorithm \ref{alg:RDA} ensures the seller's revenue as 
    \begin{align*}
        \sum_{j=1}^m \hat{p}_j(\omega^j) \ge \sum_{j=1}^m \hat{p}_j(w^*) - O\left(\sqrt{m \log(m)} \right).
    \end{align*}
\end{lemma}

Now we state our main result in this section.
\begin{theorem}\label{thm:alg}
    Algorithm \ref{alg:RDA} ensures that the seller's revenue is at least as the optimal revenue minus a sub-linear loss, i.e.,
    \begin{align*}
        \sum_{j=1}^m \hat{p}_j(\omega^j) \ge \sum_{j=1}^m \hat{p}_j(w^*) - O\left(\sqrt{m \log(m)} \right).
    \end{align*}
    Furthermore, for each buyer $i$, the difference between the optimal strategy $w^*_i$ and the one $\omega^{m+1}_i$ computed by Algorithm \ref{alg:RDA} satisfies:
    \begin{gather*}
        \left|w^*_i-\omega^{m+1}_i\right|\le O\left(\sqrt{\frac{\log m}{m}}\right).
    \end{gather*}

    Each buyer suffers a regret:
    \begin{gather*}
        \left|u_i^m - u_i^*\right| \le O\left(\sqrt{m\log(m)}\right).
    \end{gather*}
\end{theorem}
Theorem \ref{thm:alg} follows directly from  Lemma \ref{lemma:strategy_comparison}, Lemma \ref{lemma:utility_bound}, and Lemma \ref{lemma:seller_revenue}, whose proofs are deferred to the Appendix.

\begin{remark}[Comparison with Online Learning Algorithms in First-Price Auctions]
\citet{aggarwal2025no} and \citet{deng2025no} also investigate online learning in first-price auctions, but under a fundamentally different formulation. They focus on a single-buyer stochastic setting, where both the buyer's value and the highest competing bid are drawn from joint, identical, and independent distributions (i.i.d.) for each round. In contrast, we consider a multiple-buyer setting where all buyers learn their strategies simultaneously. 
While prior works formulate the single-buyer problem as a standard optimization program solved via a primal-dual framework, we formulate a global program capturing the joint strategies of all buyers, showing that its solution corresponds to an equilibrium of first-price auctions. Leveraging this global equilibrium structure allows our algorithm to achieve a regret of $\tilde{O}(\sqrt{m})$ over $m$ auctions under bandit feedback, improving upon the $O\left(m^{\frac{3}{4}}\right)$ bound established in those works.
\end{remark}

\bibliographystyle{ecta}
\bibliography{sample-bibliography.bib}

\section{Appendix}
\subsection{Example 1}
\begin{example}\label{example:tie_breaking_rule}
Consider that there are only two buyers with only budget constraints. The buyers both have a budget of 1. Suppose there is an item with a price of 1, and the buyers both have positive valuations for the item. Then buyer 1 and buyer 2 both want to buy this item. The seller can increase the price of the item. As the price increases, the demand of each buyer decreases continuously. The total demand equals the supply when the price is 2, as each buyer can only afford to buy 0.5 units. This results in fractional allocation $x_1 = x_2 = 0.5$. Under this allocation, each buyer maximizes their utility.

Now consider the case with RoS constraints. Suppose there are also two buyers, both with sufficiently large budgets. Buyer 1 has a RoS constraint of 1.2, and Buyer 2 has an RoS constraint of 1.5. Consider an item priced at 3, where buyer 1 has a valuation of 2.5, and buyer 2 has a valuation of 2 for this item. At this price, both buyers are willing to purchase the item. However, in this case, the seller cannot increase the price further, as it would exceed the buyer's allowable spending under their RoS constraints, resulting in no buyers willing to purchase. This means that the demand for a RoS binding buyer would suddenly drop from 1 to 0 as the price increases. Moreover, each buyer cannot maximize their utility with any fractional allocation. In fact, in the example above, fractional allocations make no sense, as each buyer would simply buy 1 unit of the item if its value-to-price ratio exceeds the RoS of the buyer, and buy nothing otherwise.
\end{example}
\subsection{Example 2}
\begin{example}[Tightness of the Bound]\label{ex:revenue_tight}
    Consider a setting with two buyers and two items. Both buyers have a target RoS of $\tau_1 = \tau_2 = 1$. Buyer 1 and buyer 2 both possess a budget of $\lambda_1 = \lambda_2 = 1$. The value profiles are constructed as follows for a small parameter $\epsilon > 0$, shown in Table \ref{table:tight_bound}.

\begin{table}[h]
\centering % 使表格居中
\caption{Buyers' value profile} % 描述性标题
\label{tab:corrected_data}
\begin{tabular}{l c c c} % 'l' 表示买家左对齐，'c' 表示数据居中
\toprule % 顶部粗线
 & \textbf{Item 1} & \quad & \textbf{Item 2} \\ 
\midrule % 中间线，在表头下方
Buyer 1 & 1  & \quad & $\frac{1}{\epsilon}$ \\ 
Buyer 2 & 0 & \quad & $1-\epsilon$ \\ % 修正：将原Buyer 1更改为Buyer 2
\bottomrule % 底部粗线
\end{tabular}
\label{table:tight_bound}
\end{table}

\begin{itemize}
    \item \textbf{The Mechanism Solution (Program (\ref{prog: dual}))} For any $\epsilon > 0$, the dual variables corresponding to the budget and RoS constraints are
    \begin{gather*}
        w_1^*  = \frac{\epsilon}{1 + \epsilon}, w_2^* = 1.
    \end{gather*}
    Consequently, the mechanism allocates both items to buyer 1:
    \begin{gather*}
        x_{11}^* = 1, x_{12}^*=1, x_{21}^* = 0, x_{22}^* = 0.
    \end{gather*}
    Since buyer 1 wins both items, the platform's total revenue is:
    \begin{gather*}
        \text{Rev}^* =  1.
    \end{gather*}

    \item \textbf{The First-best Solution:} In the first-best (FB) scenario, the optimal parameters are:
    \begin{gather*}
        \alpha_1^{FB} = 1, \alpha_2^{FB} = 0, \beta_1^{FB} = \frac{\epsilon}{2}, \beta_2^{FB} = 1.
    \end{gather*}
    Thus, the first best solution allocates item 1 to buyer 1 and item 2 to buyer 2,
    \begin{gather*}
        x_{11}^{FB} = 1, x_{12}^{FB} = 0, x_{21}^{FB} = 0, x_{22}^{FB} = 1.
    \end{gather*}
    The total revenue achievable in the first best case is 
    \begin{gather*}
        \text{Rev}^{FB} = 2 - \epsilon.
    \end{gather*}

    \item Comparing the two outcomes, we have:
    \begin{gather*}
        \text{Rev}^{FB} = (2 - \epsilon)\text{Rev}^*.
    \end{gather*}
    Since for any $\epsilon > 0$, the bound holds, demonstrating that the theoretical bound derived previously is indeed tight.
\end{itemize}

\end{example}

\subsection{Omitted Proofs of Section \ref{sec:theory}}

\subsubsection{Proof of Lemma \ref{lemma:feasible}}
\begin{proof}
	We first show the RoS constraint of each buyer is satisfied. 
    % It is easy to check that both Program \eqref{prog: primal} and \eqref{prog: dual} are convex. Thus the complementary slackness conditions hold, and we have:
	% \begin{align*}
	% 	x^*_{i,j}(w^*_{i}v_{i,j}-p^*_j) = 0, \forall i \in [n], \forall j \in [m].
	% \end{align*}
    Consider the complementary slackness condition in Equation \eqref{eq:sop_cs_x}.
	Summing over all $m$ items yields:
	\begin{align}\label{eq:lemma_1} 
		w^*_i\sum_{j=1}^m  v_{i,j}x^*_{i,j} = \sum_{j=1}^mp^*_jx^*_{i,j}, \forall i \in [n].
	\end{align}
	The above equation establishes a direct relationship between the total value of allocated items and the payment of each buyer.
	The second constraint in Program \eqref{prog: dual} implies:
	% From the second constraint in the dual program (\ref{prog: dual}), we have 
	% \begin{align*}
	% 	\beta_i \le \frac{1}{\tau_i}, \forall i \in [n].
	% \end{align*}
	% Substituting this inequality into the above summation:
	\begin{align*}
        \sum_{j=1}^m p^*_j x^*_{i,j} = w^*_i\sum_{j=1}^m  v_{i,j}x^*_{i,j} \le \frac{1}{\tau_i} \sum_{j=1}^m v_{i,j}x^*_{i,j},
	\end{align*}
	% Rearranging terms, we obtain
    which is equivalent to $\sum_{j=1}^m v_{i,j}x^*_{i,j}\ge \tau_i \sum_{j=1}^m p^*_j x^*_{i,j}$.
	% \begin{align*}
	% 	\frac{\sum_{j=1}^m v_{i,j}x_{i,j}}{\sum_{j=1}^m p_j x_{i,j}} \ge \tau_i.
	% \end{align*}
	This indicates that buyer $i$'s RoS constraint is satisfied.
	
	Now we show the budget constraint is satisfied. 
    Combining Equation \eqref{eq:sop_cs_beta} and \eqref{eq:lemma_1}, we get:
	\begin{align}\label{eq:lemma1_2}
		\sum_{j=1}^m p^*_j x^*_{i,j} = w^*_i (c^*_i - d^*_i).
	\end{align}
 %    Still using the complementary slackness condition, we have
	% \begin{gather} 
	% 	d^*_i\left(w^*_i - \frac{1}{\tau_i}\right) = 0, \forall i \in [n].\label{eq:cs_t}\\
 %        w^*_i\left(c_i - \sum_{j=1}^m x^*_{i,j}v_{i,j} - d^*_i\right) = 0, \forall i \in [n].\label{eq:cs_u}
	% \end{gather}
	The KKT stationarity condition gives:
	\begin{align}\label{eq:kkt_u}
		w^*_i c^*_i = \lambda_i,
	\end{align}
%	\begin{align}\label{eq:kkt_x}
%		\beta_i v_{ij} = p_j.	
%	\end{align}
%	Multiplying $x_{ij}$ for equation (\ref{eq:kkt_x}), we have that
%	\begin{align*}
%		\beta_i x_{ij} v_{ij} = p_j x_{ij}.
%	\end{align*}
%	Summation over all $m$ items, we have
%	\begin{align} \label{eq:lemma1_1}
%		\beta_i \sum_{j=1}^m x_{ij}v_{ij} = \sum_{j=1}^m p_j x_{ij}
%	\end{align}
	Plugging in Equation \eqref{eq:kkt_u} and \eqref{eq:sop_cs_t}, we finally have 
	\begin{align}\label{eq:cap_payment}
		\sum_{j=1}^m p^*_j x^*_{i,j} = \lambda_i - \frac{d^*_i}{\tau_i}\le \lambda_i,
	\end{align}
	where the last inequality follows from the fact that $d^*_i \ge 0$.
	This confirms that buyer $i$'s budget constraint is satisfied.
\end{proof}

\subsubsection{Proof of Lemma \ref{lemma:beta_lower_bound}}
\begin{proof}
    For any buyer $i$, $w_i^*$ must satisfy the feasibility condition $w_i^* \le \frac{1}{\tau_i}$. We consider two cases:
    \begin{itemize}
        \item $w_i^* = \frac{1}{\tau_i}$. In this case, $w_i^* \ge \frac{1}{\tau_i}$ holds;
        \item $w_i^* \neq \frac{1}{\tau_i}$. In this case, according to the complementary slackness condition (\eqref{eq:sop_cs_t}), we must have $d_i^* = 0$. 
        Combining Equation (\eqref{eq:lemma_1}) and (\eqref{eq:cap_payment}) gives
        \begin{gather*}
            w_i^* = \frac{\lambda_i - \frac{d_i^*}{\tau_i}}{\sum_{j=1}^m v_{i,j} x_{i,j}^*} = \frac{\lambda_i}{\sum_{j=1}^m v_{i,j}x_{i,j}^*} \ge \frac{\lambda_i}{m \bar{v}}
        \end{gather*}
    \end{itemize}
    One of the above two cases must hold, and thus we have:
    \begin{gather*}
        w_i^* \ge \min \left\{\frac{\lambda_i}{m \bar{v}}, \frac{1}{\tau_i} \right\} = \underline{w}_i, \forall i\in [n].
    \end{gather*}
\end{proof}

\subsubsection{Proof of Lemma \ref{lemma:SOP_IOP}}
\begin{proof}
    The dual of Program \eqref{IOP:primal} is:
	\begin{mini}|l|[2]{r_i,\alpha_i,\mu_{i, \cdot}}{\sum_{j=1}^m q_{i,j}\mu_{i,j} + \lambda_i r_i}{\tag{I-DUAL}\label{IOP:dual}}{}
	\addConstraint{(1+\alpha_i)v_{i,j} - (r_i + \alpha_i \tau_i)p^*_j - \mu_{i,j}}{\le 0}{\quad \forall j \in [m]}
	\addConstraint{\mu_{i,j}}{\ge 0 }{\quad \forall j\in [m]}
	\addConstraint{\alpha_i}{\ge 0}{}
	\addConstraint{r_i}{\ge 0.}{}
	\end{mini}
    
	We prove Lemma \ref{lemma:SOP_IOP} by constructing solutions to Program \eqref{IOP:primal} and \eqref{IOP:dual}, and showing that these solutions lead to the same objective of the two programs. 
    Then they are both optimal as they satisfy the strong duality condition. We first set $x_{i,j}=x^*_{i,j},\forall j \in [m]$. As for other variables, we consider the following two cases:
    % Recall that the solution $x_{i,j}^* = \{0, 1\}, \forall i\in [n], \forall j \in [m]$.
\begin{enumerate}
	\item If $d^*_i > 0$, we have $w^*_i = \frac{1}{\tau_i}$ according to Equation \eqref{eq:sop_cs_t}. %Based on Equation \eqref{eq:cap_payment}, we know that buyer $i$'s payment is less than his budget.
    % In this case, we must have $x_{i,j}^* \in\{0,1\},\forall j$, since otherwise, the buyer $i$'s payment is less than his budget. Then if $x_{i,j}^* < 1$ and $x_{i,j}^* > 0$, then improving $x_{i,j}$ results in a higher utility for buyer $i$.
    In this case, we set:
	% \begin{gather*}
 %        y_i = 0,\\
	% 	\alpha_i = \frac{\sum_{j=1}^m v_{i,j}x^*_{i,j} \left[ \sum_{j=1}^m x^*_{i,j}v_{i,j} + d^*_i\right]}{\lambda_i \tau_i + (\tau_i - 1) \left\{\sum_{j=1}^m x^*_{i,j} v_{i,j}\left[\sum_{j=1}^m x^*_{i,j}v_{i,j} + t^*_i \right] \right\}},\\
	% 	\mu_{i,j} = \max \left\{ \max_{j \in [m]}\left\{(1 + \alpha_i)v_{i,j} - (\alpha_i \tau_i)p^*_j\right\},0\right\},\forall j.
	% \end{gather*}
        \begin{gather*}
            r^*_i = 0, \\
            \alpha^*_i = \max_{j: \tau_i p_j^* \neq v_{i, j}} \frac{v_{i, j}}{\tau_i p_j^* - v_{i, j}},\\
            \mu^*_{i, j} = \begin{cases}
                v_{i, j} &\text{if } x_{i, j}^* > 0\\
                0 &\text{otherwise}
            \end{cases}.
        \end{gather*}
        % \begin{align*}
        %     \mu_{i, j} = \begin{cases}
        %         v_{i, j} &\text{if } x_{i, j}^* = 1\\
        %         0 &\text{otherwise}.
        %     \end{cases}
        % \end{align*}
        We first show that the above constructed solution $(r^*_i,\alpha^*_i,\mu^*_{i,\cdot})$ is a feasible solution to the dual program.
        It is obvious that $\mu^*_{i, j} \ge 0, \forall j \in [m]$ and $y^*_i \ge 0$.
        Since $p_j^* = \max_{i} \{w_i^* v_{i, j}\}$ and $w_i^* \le \frac{1}{\tau_i}$, we have $v_{i, j} \le \tau_i p_j^*, \forall j \in [m]$.
        % \begin{gather*}
        %     v_{i, j} \le \tau_i p_j^*, \forall j \in [m].
        % \end{gather*}
        It follows that $\alpha^*_i \ge 0$.
        % \begin{gather*}
        %     \alpha_i \ge 0
        % \end{gather*}

        If $x_{i, j}^* > 0$, we have $\tau_i p_j^* = v_{i, j}$ according to Equation \eqref{eq:sop_cs_x}. Therefore,
        \begin{align*}
            (1+\alpha^*_i)v_{i,j} - (r_i^* + \alpha^*_i \tau_i)p^*_j - \mu^*_{i,j}
            =v_{i,j} - \mu^*_{i,j} + \alpha^*_i (v_{i,j} - \tau_i p_j^*) - r^*_i p_j^*
            = 0.
        \end{align*}
        If $x_{i,j}^* = 0$, since $\alpha^*_i\ge \frac{v_{i,j}}{\tau_i p_j^* - v_{i,j}}$ by definition, we have
        \begin{align*}
            (1 + \alpha^*_i) v_{i,j} - (r^*_i + \alpha^*_i \tau_i)p_j^* - \mu^*_{i,j} = v_{i,j} + \alpha^*_i(v_{i,j} - \tau_i p_j^*)\le 0.
        \end{align*}
        In both cases, the first constraint of Program \eqref{IOP:dual} is satisfied.
        % Since 
        % \begin{align*}
        %     \frac{v_{i,j}}{\tau_i p_j^* - v_{i,j}} - \alpha_i \le 0
        % \end{align*}
        % Thus we have $v_{i,j} + \alpha_i(v_{i,j} - \tau_i p_j^*) \le 0$

        It is easy to check that under the above construction, the primal objective and the dual objective are the same, which implies that both programs are optimal according to the strong duality condition.
	% It is easy to check that $(x_i, y_i, \alpha_i, \mu_{i,j})$ satisfies the complementary slackness conditions (\ref{eq:iop_cs_y}), (\ref{eq:iop_cs_alpha}), (\ref{eq:iop_cs_mu}), (\ref{eq:iop_cs_x}).
	\item If $d^*_i = 0$, we set $\alpha^*_i = 0$, $r^*_i=\frac{1}{w^*_i}$, and $\mu^*_{i,j} = 0,\forall j$.
        % Following the same induction, we first prove the dual program is feasible.
        It is clear that $\alpha_i^* \ge 0, r_i^* \ge 0, \mu_{i,j}^* \ge 0, \forall j$. We also have
        \begin{gather*}
            (1+\alpha^*_i)v_{i,j} - (r^*_i + \alpha^*_i \tau_i)p^*_j - \mu^*_{i,j} = v_{i,j} - \frac{1}{w^*_i} p^*_j \le 0,
        \end{gather*}
        where the last inequality is due to $p_j^* = \max_{i} \{w_i^* v_{i, j}\}$. Combining Equation \eqref{eq:lemma_1} and \eqref{eq:cap_payment}, we get
        \begin{align*}
            w^*_i \sum_{j=1}^m v_{i,j} x_{i, j}^* = \lambda_i-\frac{d^*_i}{\tau_i}=\lambda_i.
        \end{align*}
        It follows that
        \begin{align*}
            \sum_{j=1}^m v_{i,j} x_{i,j}^* = \sum_{j=1}^m x_{i,j}^*\mu^*_{i,j} + r^*_i \lambda_i.
        \end{align*}
        This means that the primal objective equals the dual objective, hence both are optimal.
\end{enumerate}

\end{proof}

\subsubsection{Proof of Lemma \ref{lemma:competitive equilibrium}}
\begin{proof}
    Lemma \ref{lemma:SOP_IOP} already shows that with prices $p^*$, each buyer's best choice, under both their budget and RoS constraints, is $x^*_{i,\cdot}$. Since $p^*_j>0,\forall j$, we also have $\sum_{i=1}^n x^*_{i,j}=1, \forall j$ according to Equation \eqref{eq:sop_cs_p}, which implies market clearance.
    % Therefore, to prove Lemma \ref{lemma:competitive equilibrium}, it suffices to show that the prices $p^*$ clears the market, i.e., $\sum_{i=1}^nx_{i,j}=1,\forall j$.
    
\end{proof}

\subsubsection{Proof of Lemma \ref{lemma:IOP_SOP}}
\begin{proof}
    
    With Assumption \ref{assumption:no_buyer_win_all}, we claim that under competitive equilibrium $(x^*, p^*)$, each buyer has either a binding budget constraint or a binding RoS constraint, i.e., at least one of the following conditions must hold,
    \begin{gather*}
        \sum_{j=1}^m p_j^* x_{i,j}^* = \lambda_i, \\
        \sum_{j=1}^m v_{i,j} x_{i,j}^* =  \tau_i \sum_{j=1}^m p_j^* x_{i,j}^*.
    \end{gather*}
    We prove this by contradiction. Suppose there exists a buyer $i$, such that neither of the two constraints is binding:
    \begin{gather*}
        \sum_{j=1}^m p_j^* x_{i,j}^* < \lambda_i, \\
        \sum_{j=1}^m v_{i,j} x_{i,j}^* >  \tau_i \sum_{j=1}^m p_j^* x_{i,j}^*.
    \end{gather*}
    According to Assumption \ref{assumption:no_buyer_win_all}, buyer $i$ cannot buy all the items. So there exists item $j$ with $x^*_{i, j} < 1$. In this case, buyer $i$ can slightly increase $x^*_{i,j}$ to improve their utility while at the same time satisfying both the above two constraints. This clearly contradicts the fact that  $(x^*, p^*)$ is a competitive equilibrium.

    Then, similar to Lemma \ref{lemma:SOP_IOP}, we also prove the Lemma \ref{lemma:IOP_SOP} by explicitly constructing solutions to Program \eqref{prog: primal} and \eqref{prog: dual} and applying strong duality to show that they are both optimal.
    
    % Since $(x^*,p^*)$ is a competitive equilibrium, we have that $x^*_{i,\cdot}$ is an optimal solution to Program \eqref{IOP:primal} if $p=p^*$. It also implies market clearance, i.e., $\sum_{i=1}^nx_{i,j}^* = 1$ for any item $j$, satisfying the feasibility constraint of Program \eqref{prog: dual}.
    % For each buyer $i$, we know that $\sum_{i=1}^nx_{i,j}^* = 1, \forall j\in[m]$, meaning that $x^*$ satisfies the feasible constraint of CAP.
    For each buyer $i$, we construct the solution as follows:
    \begin{gather*}
        d^*_i = \begin{cases}
            0 & \text{if } \sum_{j=1}^m p_j^*x_{i,j}^* = \lambda_i \\
            \lambda_i \tau_i - \sum_{j=1}^m v_{i,j}x_{i,j}^* & \text{otherwise}
        \end{cases},\\
        c^*_i=\sum_{j=1}^mx^*_{i,j}v_{i,j}+d^*_i,\\
        w^*_i = \begin{cases}
            \frac{\sum_{j=1}^m p_j^*x_{i,j}^*}{\sum_{j=1}^m v_{i,j}x_{i,j}^*} & \text{if }\sum_{j=1}^m p_j^*x_{i,j}^* = \lambda_i\\
            \frac{1}{\tau_i} & \text{otherwise}
        \end{cases}.
    \end{gather*}
    Note that in the above construction, when $\sum_{j=1}^m p_j^*x_{i,j}^* = \lambda_i$, we must have $\sum_{j=1}^m v_{i,j}x_{i,j}^*>0$. Otherwise, if $\sum_{j=1}^m v_{i,j}x_{i,j}^*=0$, then $\sum_{j=1}^m v_{i,j} x_{i,j}^* < \tau_i \sum_{j=1}^m p_j^* x_{i,j}^*$, violating the RoS constraint.

    We first show that $(x^*,d^*,c^*)$ is a feasible solution to Program \eqref{prog: primal}. It suffices to check the third constraint as other ones are straightforward. Consider two cases.
    If $\sum_{j=1}^m p_j^*x_{i,j}^* = \lambda_i$, i.e., the budget constraint is binding, we have $d^*_i = 0 \ge 0$.
    And if $\sum_{j=1}^m p_j^*x_{i,j}^* < \lambda_i$, we know that the RoS constraint must be binding according to the analysis above, i.e., $\tau_i\sum_{j=1}^m p_j^*x_{i,j}^* = \sum_{j=1}^m v_{i,j}x_{i,j}^*$. Therefore,
    % Since $(x^*, p^*)$ is a competitive equilibrium of the auction market. It follows that given $p^*$, $x^*$ is the optimal solution of each buyer's program \eqref{IOP:primal}. Then let $\mu^*, y^*, \alpha^*$ be the optimal solution of each buyer's dual program \eqref{IOP:dual}. Based on the strong duality, we have for each buyer, $i \in [n]$, 
    % \begin{gather*}
    %     \sum_{j=1}^m v_{i, j}x_{i,j}^* = \sum_{j=1}^m \mu_{i,j}^* + y_i^* \lambda_i.
    % \end{gather*}
    % According to the complementary slackness condition, we also have
    % \begin{gather}
    %     y_i\left(\sum_{j=1}^m p_j^* x_{i,j}^* - \lambda_i \right) = 0 \label{IOP:cs_y} \\
    %     \alpha_i \left( sum_{}\right)
    % \end{gather}
    % \begin{gather} \label{IOP}
    % \end{gather}
    \begin{align*}
        d^*_i = \lambda_i \tau_i - \sum_{j=1}^m v_{i,j} x_{i,j}^* 
        = \tau_i \left[\lambda_i - \frac{1}{\tau_i} \sum_{j=1}^m v_{i,j} x_{i,j}^*\right] 
        = \tau_i \left[\lambda_i - \sum_{j=1}^m p_j^* x_{i,j}^* \right] > 0.
    \end{align*}
    % Thus, the primal solutions of $d_i$ are feasible.

    Then, we show that $(p^*,w^*)$ is a feasible solution to Program \eqref{prog: dual}.
    For any buyer $i$, if $\sum_{j=1}^m p^*_{j}x_{i,j}^* = \lambda_i$, then for any item $j$, we have
    \begin{align*}
        p_{j}^*-w^*_i v_{i,j} &= p_{j}^* - v_{i, j} \frac{\sum_{j'=1}^m p_{j'}^*x_{i,j'}^*}{\sum_{j=1}^m v_{i,j'}x_{i,j'}^*} \\
        &= p_{j}^* - v_{i,j}\frac{p_{j}^* x_{i,j}^*}{\sum_{j'=1}^m v_{i,j'}x_{i,j'}^*} - v_{i,j}\frac{\sum_{j' \neq j} p_{j'}^*x_{i,j'}^*}{\sum_{j'=1}^m v_{i,j'}x_{i,j'}^*} \\
        &\ge p_{j}^*\left(1 - \frac{v_{i,j} x_{i,j}^*}{\sum_{j'=1}^m v_{i,j'}x_{i,j'}^*} \right) \\
        &\ge 0.
    \end{align*}
    Otherwise, we have $w^*_i=\frac{1}{\tau_i}$ and $p^*_j-w^*_iv_{i,j}=p^*_j-\frac{v_{i,j}}{\tau_i}$. In this case, we can only have $p^*_j\ge\frac{v_{i,j}}{\tau_i}$, since otherwise, buyer $i$ can get a strictly higher utility by marginally increase $x^*_{i,j}$ without violating the budget and RoS constraints.
    % If $\sum_{j=1}^m v_{i,j}x_{i,j}^* = 0$, then for each item $j$, either $v_{i,j}=0$ or $x^*_{i,j}=0$. If $v_{i,j}=0$, we clearly have $p_j^*-w^*_i v_{i,j}=p^*_j> 0$. And if $v_{i,j}>0$, we must have $x^*_{i,j}=0$, thus $p^*_j-w^*_iv_{i,j}=p^*_j-\frac{v_{i,j}}{\tau_i}$ according to the construction of $w^*_i$. In this case, we can only have $p^*_j\ge\frac{v_{i,j}}{\tau_i}$, since otherwise, buyer $i$ can get a strictly higher utility by marginally increase $x^*_{i,j}$ without violating the budget and RoS constraints.
    
    % For any item $k \in [m]$ any buyer $i \in [n]$, if $\sum_{j=1}^m v_{i,j}x_{i,j}^* = 0$, meaning that $x_{i,k}^* = 0$, then we have
    % \begin{align*}
    %     p_k^*-w_i v_{i,k} &= p_k^* - \frac{v_{i,k}}{\tau_i} > 0.
    % \end{align*}
    % Otherwise, buyer $i$ buys item $k$ resulting in a higher utility, while not violating his RoS constraint.

    Finally, we show that the primal objective equals the dual objective, and based on the strong duality, we conclude that $x_{i,j}^*$ and $p_j^*$ are optimal solutions to Program \eqref{prog: primal} and \eqref{prog: dual}, respectively.
    Note that for each buyer $i$, we consider the following two cases:
    \begin{itemize}
        \item If $\sum_{j=1}^m p_j^* x_{i,j}^* = \lambda_i$, we have
        \begin{align*}
            \sum_{j=1}^m v_{i,j}x_{i, j}^* + d^*_i = \sum_{j=1}^m v_{i,j}x_{i, j}^* = \frac{\sum_{j=1}^m p_j^* x_{i,j}^*}{w^*_i} = \frac{\lambda_i}{w^*_i},
        \end{align*}
        and
        \begin{align*}
            \frac{d^*_i}{\tau_i} = 0 = \lambda_i - \sum_{j=1}^m p_j^* x_{i,j}^*.
        \end{align*}
        \item If $\sum_{j=1}^m p_j^* x_{i, j}^* < \lambda_i$,  we have
        \begin{align*}
            \sum_{j=1}^m v_{i,j}x_{i, j}^* + d^*_i = \lambda_i \tau_i = \frac{\lambda_i}{w^*_i},
        \end{align*}
        and
        \begin{align*}
            \frac{d^*_i}{\tau_i} = \frac{\lambda_i \tau_i - \sum_{j=1}^m v_{i,j}x_{i,j}^*}{\tau_i} = \lambda_i - \frac{\sum_{j=1}^m v_{i,j}x_{i,j}^*}{\tau_i}.
        \end{align*}
        In this case, the budget constraint is not binding. According to the above analysis, the RoS constraint must be binding, i.e., $\tau_i\sum_{j=1}^m p_j^*x_{i,j}^* = \sum_{j=1}^m v_{i,j}x_{i,j}^*$. It follows that 
        \begin{gather*}
            \frac{d^*_i}{\tau_i}=\lambda_i - \sum_{j=1}^m p_j^*x_{i,j}^*.
        \end{gather*}
    \end{itemize}
    In either case, we have
    \begin{align*}
        \sum_{i=1}^n \left[\lambda_i \log\left(\sum_{j=1}^m v_{i,j}x_{i,j}^* + d^*_i\right) - \frac{d^*_i}{\tau_i}\right] &= \sum_{i=1}^n \left[\lambda_i \log\left(\frac{\lambda_i}{w^*_i} \right) -  \lambda_i + \sum_{j=1}^m p_j^* x_{i,j}^*\right] \\
        &=\sum_{i=1}^n \left[\lambda_i \log(\lambda_i) - \lambda_i -\lambda_i \log(w^*_i) \right] + \sum_{j=1}^m p_j^* \left( \sum_{i=1}^n x_{i,j}^*\right) \\
        &=\sum_{i=1}^n \lambda_i\left[ \log(\lambda_i) - 1\right] -\sum_{i=1}^n\lambda_i \log(w^*_i)  + \sum_{j=1}^m p_j^*.
    \end{align*}
    This indicates that under our construction, the primal objective is the same as the dual objective. Consequently, both solutions are optimal.
\end{proof}

\subsubsection{Proof of Lemma \ref{lemma:strong_convex}}
\begin{proof}
    Observe that
    \begin{align*}
        \nabla^2 f(w) = \left[\begin{matrix}
            \frac{\lambda_1}{w_1^2} & 0 & \cdots & 0 \\
            0 & \frac{\lambda_2}{w_2^2} & \cdots & 0 \\
            \vdots & \vdots & \ddots & \vdots \\
            0 & 0 & \dots & \frac{\lambda_n}{w_n^2}
        \end{matrix} \right] \succeq \sigma \left[\begin{matrix}
            1 & 0 & \cdots & 0 \\
            0 & 1 & \cdots & 0 \\
            \vdots & \vdots & \ddots & \vdots \\
            0 & 0 & \dots & 1
        \end{matrix} \right].
    \end{align*}
    The statement in Lemma \ref{lemma:strong_convex} immediately follows.
\end{proof}

\subsubsection{Proof of Lemma \ref{lemma:unique}}
\begin{proof}

We prove this lemma in two parts: first establishing price uniqueness, then showing that despite potential non-unique allocations, all equilibria generate the same revenue.

Based on Lemma \ref{lemma:strong_convex}, we know that the objective function of Program \eqref{prog: dual} is strongly convex since the addition term $\sum_{j=1}^mp_j$ is linear, which makes Program \eqref{prog: dual} a strongly convex program. This immediately implies that the optimal solution $(p^*, w^*)$ is unique. 

While the prices are unique, the primal allocations may not be unique due to ties in buyers' weighted valuations. 
We illustrate this with an example and then prove revenue uniqueness.

Consider a situation where there exists a solution $(\hat{x}, \hat{d},\hat{c})$ with two RoS constrained buyers $i, k$, i.e., $\hat{d}_i > 0$ and $\hat{d}_k > 0$. This means

    \begin{align*}
        w_i^* = \frac{1}{\tau_i}, w_k^* = \frac{1}{\tau_k}.
    \end{align*}

Suppose for some item $j$, the weighted valuations for buyer $i$ and $k$ tie at the optimal price:
\begin{align*}
         p_j^* = \frac{v_{i,j}}{\tau_i} = \frac{v_{k,j}}{\tau_k} > w_{i'}^*v_{i',j}, \forall i' \neq i, i' \neq k.
 \end{align*}

In this case, the complementary slackness condition (Equation \eqref{eq:sop_cs_x}) implies $\hat{x}_{i',j} = 0$ for all $i' \neq i,k$. However, both buyer $i$ and buyer $k$ have the same weighted valuation equal to $p_j^*$, so the item can be allocated to either buyer (or split fractionally between them).

Specifically, consider two different solutions:

\begin{itemize}
        \item \textbf{Solution 1:} Allocate item $j$ to buyer $k$:
        \begin{align*}
            \hat{x}_{i,j} = 0, \quad \hat{x}_{k,j} = 1, \quad \hat{x}_{i',j} = 0 \text{ for } i' \neq i,k.
        \end{align*}
        Set $\hat{d}_i = \lambda_i \tau_i - \sum_{j'=1}^m v_{i,j'}\hat{x}_{i,j'}$ and $\hat{d}_k = \lambda_k \tau_k - \sum_{j'=1}^m v_{k,j'}\hat{x}_{k,j'}$.

        \item \textbf{Solution 2:} Allocate item $j$ to buyer $i$:
        \begin{align*}
            \tilde{x}_{i,j} = 1, \quad \tilde{x}_{k,j} = 0, \quad \tilde{x}_{i',j} = 0 \text{ for } i' \neq i,k.
        \end{align*}
        Set $\tilde{d}_i = \lambda_i \tau_i - \sum_{j'=1}^m v_{i,j'}\tilde{x}_{i,j'}$ and $\tilde{d}_k = \lambda_k \tau_k - \sum_{j'=1}^m v_{k,j'}\tilde{x}_{k,j'}$.
    \end{itemize}

    Both solutions satisfy the complementary slackness conditions and are therefore optimal. However, we now show that despite different allocations, the total revenue remains the same.

    The total revenue in any competitive equilibrium is given by:
    \begin{align*}
        \text{Revenue} = \sum_{j=1}^m p_j^* \left(\sum_{i=1}^n x_{i,j}\right) = \sum_{j=1}^m p_j^*,
    \end{align*}
    where the second equality follows from the market clearing condition $\sum_{i=1}^n x_{i,j} = 1$ for all items $j$ with $p_j^* > 0$ (by the complementary slackness condition in Equation \eqref{eq:sop_cs_p}).

    Since the prices $p^*$ are uniquely determined, and the revenue depends only on the sum $\sum_{j=1}^m p_j^*$, all competitive equilibria generate the same revenue, regardless of how items are allocated when ties occur. This completes the proof.

    % Based on Lemma \ref{lemma:strong_convex}, we know that the objective function of Program \eqref{prog: dual} is strongly convex since the addition term $\sum_{j=1}^mp_j$ is linear, which makes Program \eqref{prog: dual} a strongly convex program. This immediately implies that the optimal solution $(p^*, w^*)$ is unique. 
    % However, the primal allocations are not unique. The reason is as follows.
    % Suppose there are one solutions $(\hat{x}, \hat{d})$ and two buyers with $\hat{d}_k \neq 0, \hat{d}_i \neq 0$, meaning
    % \begin{align*}
    %     \beta_i^* = \frac{1}{\tau_i}, \beta_k^* &= \frac{1}{\tau_k}.
    % \end{align*}

    % For item $j$, if 
    % \begin{align*}
    %     p_j^* = \frac{v_{i,j}}{\tau_i} = \frac{v_{k,j}}{\tau_k} > \beta_{i'}^*v_{i',j}, \forall i' \neq i, i' \neq k.
    % \end{align*}
    % Thus, for $i' \neq i,k$, $\hat{x}_{i',j} = 0$.
    % Now consider 
    % $\hat{x}_{i,j} = 0, \hat{d}_i = \lambda_i \tau_i - \sum_{j=1}^m v_{i,j}\hat{x}_{i,j}$
    % $\hat{x}_{k,j} = 1, \hat{d}_k = \lambda_k \tau_k - \sum_{j=1}^m v_{k,j}\hat{x}_{k,j}$
    % Both satisfy the complementary slackness condition.
    % Meaning this is an optimal solution.

    % Now consider another solution $\tilde{x}$, with $\tilde{x}_{i,j} = 1, \tilde{d}_i = \lambda_i \tau_i - \sum_{j=1}^m v_{i,j} \tilde{x}_{i,j}$, and 
    % $\tilde{x}_{k,j}=0, \tilde{d}_k = \lambda_k \tau_k - \sum_{j=1}^m v_{k,j}\tilde{x}_{k,j}$. Both satisfy the complementary slackness condition, meaning this is an optimal solution. 
    % However, under different solutions, the revenue is the same.
\end{proof}

\subsubsection{Proof of Lemma \ref{lemma:Pareto_efficiency}}
\begin{proof}
    % Assume the optimal allocation computed by the program (\ref{prog: primal}) is $x$ and each buyer, $i\in [n]$, has utility $u_i = \sum_{j=1}^m v_i^j x_i^j$. Also, we assume $t$ is the optimal solution.
    Assume, on the contrary, that there exists another feasible allocation $\hat{x}$, such that each buyer $i$ has utility $\hat{u}_i \ge u_i$ and there exists buyer $i$ with $\hat{u}_i > u_i$. Since $\hat x$ is feasible, we must have $\hat x_{i,j}\ge 0$ and $\sum_{i=1}^n\hat x_{i,j}\le 1$. Define $\hat c_i=\sum_{j=1}^m\hat x_{i,j}v_{i,j}+d^*_i,\forall i$. Then one can easily verify that $(\hat x, d^*, \hat c)$ is also a feasible solution to Program \eqref{prog: primal}.
    However, we have
    \begin{align*}
        \sum_{i=1}^n \lambda_i \log(\hat c_i) - \frac{d^*_i}{\tau_i} =\sum_{i=1}^n \lambda_i \log(\hat{u}_i + d^*_i) - \frac{d^*_i}{\tau_i} > \sum_{i=1}^n\lambda_i \log(u_i + d^*_i) - \frac{d^*_i}{\tau_i}=\sum_{i=1}^n\lambda_i \log(c^*_i) - \frac{d^*_i}{\tau_i},
    \end{align*}
    which indicates that the solution $(\hat x, d^*, \hat c)$ achieves a strictly higher objective, contradicting the optimality of the solution $(x^*, d^*, c^*)$.
\end{proof}

\subsubsection{Proof of Theorem \ref{thm:ic}}
\begin{proof}
    The market-clearing mechanism contains two steps: each buyer submits their budget and RoS constraints to the seller, and the seller computes the corresponding allocation $x$ and prices $p$ according to Program \eqref{prog: primal} and Program \eqref{prog: dual}. Therefore, to prove Theorem \ref{thm:ic}, it suffices to show that fixing other buyers reporting constraints $(\lambda_{-i}, \tau_{-i})$ truthfully, buyer $i$ maximizes their utility by reporting the financial constraints $(\lambda_i, \tau_i)$ truthfully.

    Let $(x^*, p^*)$ and dual variables $w^*$ denote the optimal solutions when all buyers report truthfully. Let $(\hat{x}, \hat{p})$ and dual variable $\hat{w}$ denote the solution when buyer $i$ reports $(\hat{\lambda}, \hat{\tau})$ while others reporting financial constraints truthfully.

    \textbf{Part 1: Under-reporting Constraints, i.e., $\hat{\lambda}_i \le \lambda_i, \hat{\tau}_i \ge \tau_i$.}

    \begin{itemize}
        \item \textbf{Step 1: Monotonicity of the Dual Variable $\hat{w}_i \le w_i^*$}.
        
        We compare the dual variable of buyer $i$ derived from the dual program (\ref{prog: dual}).
        
        \textit{Case 1: $w_i^* \le \frac{1}{\hat{\tau}_i} \le \frac{1}{\tau_i}.$} Since $(p^*, w^*)$ and $(\hat{p}, \hat{w})$ are optimal solutions to their respective dual programs, we can utilize the optimality condition. Specifically, comparing the dual objectives:
        \begin{align}\label{eq:ic_dual_1}
            \sum_{j=1}^m p_j^* - \sum_{i=1}^n \lambda_i \log(w_i^*) \le \sum_{j=1}^m \hat {p}_j - \sum_{i=1}^n \lambda_i \log (\hat{w}_i)
        \end{align}
        \begin{align}\label{eq:ic_dual_2}
            \sum_{j=1}^m \hat{p}_j - \sum_{k\neq i, k = 1}^n \lambda_k \log(\hat{w}_k) - \hat{\lambda}_i\log(\hat{w}_i) \le \sum_{j=1}^m p_j^* - \sum_{k\neq i, k = 1}^n \lambda_k \log(w_k^*) - \hat{\lambda}_i\log(w_i^*) 
        \end{align}
        Combining (\ref{eq:ic_dual_1}) and (\ref{eq:ic_dual_2}), it follows that
        \begin{align}\label{eq:ic_dual_comb}
            (\hat{\lambda}_i - \lambda_i) \log(w_i^*)\le (\hat{\lambda}_i - \lambda_i) \log(\hat{w}_i)
        \end{align}
        Since buyer $i$ under-reports the budget $\hat{\lambda}_i - \lambda_i \le 0$ and the logarithmic function is strictly increasing, inequality \eqref{eq:ic_dual_comb} implies that $\hat{w}_i \le w_i^*$.
        
        \textit{Case 2: $w_i^* > \frac{1}{\hat{\tau}_i}$}.
        By duality feasibility in the misreported instance, we immediately have $\hat{w}_i \le w_i^*$.

        In all cases, the misreporting buyer's dual multiplier decreases: $\hat{w}_i \le w_i^*$.

        \item \textbf{Step 2: Monotonicity of Prices $\hat{p}_j \le p_j^*, \forall j \in [m]$.} 
        
        Recall that the market clearing price is $\hat{p}_j  = \max_{k} \left\{\hat{w}_k v_{k,j} \right\}$ and $p_j^* = \max_{k} \left\{w_k^* v_{k,j} \right\}$. In step 1, we have established that for buyer $i$, we have $\hat{w}_i \le w_i^*$. We now show that no other buyer $k\neq i$ increases their dual multiplier, i.e., $\hat{w}_k \le w_k^*, \forall k \neq i$.

        \textit{RoS Constrained Buyers:} For any buyer $k \neq i$, where the RoS constraint is binding when all buyers reporting truthfully, $w_k^* = \frac{1}{\tau_k}$. Since other buyers' reporting constraints are not changed, thus $\hat{w}_k \le w_k^*$.

        \textit{Budget  Constrained  Buyers:} For any buyer $k \neq i$ whose budget is binding, the stationarity condition implies $w_k^* = \frac{\lambda_k}{\sum_{j=1}^m v_{k,j}x_{k,j}^*}$. Suppose, for the sake of contradiction, that buyer $k$ increases their multiplier: $\hat{w}_k > w_k^*$. Given that $\hat{w}_k = \frac{\hat{\lambda}_k}{\sum_{j=1}^m v_{k,j} \hat{x}_{k, j}}$, an increase in $\hat{w}_k$ implies a strict decrease in total utility: $\sum_{j=1}^m v_{k,j}\hat{x}_{k,j} < \sum_{j=1}^m v_{k,j}x_{k,j}^*$.

        However, consider the market dynamics: buyer $i$ has reduced their competitive pressure ($\hat{w}_i \le w_i^*$), and all other buyers $k$ are assumed to satisfy the budget constraint. In a Fisher market with gross substitute valuations, if one agent decreases their competitiveness (lowers $w_i$), the remaining agents face a strictly more favorable (or neutral) price environment. Therefore, their resulting utility cannot decrease (i.e., we must have $\sum_{j=1}^m v_{k,j}\hat{x}_{k,j} \ge \sum_{j=1}^m v_{k,j}x_{k,j}^*$).
        This contradicts the implication that $\sum_{j=1}^m v_{k,j}\hat{x}_{k,j} < \sum_{j=1}^m v_{k,j}x_{k,j}^*$. Thus, the assumption is false, and we must have $\hat{w}_k \le w_k^*$.

        Consequently, $\hat{w}_k \le w_k^*$ for all $k\in [n]$. And the market clearing price must be non-increasing: $\hat{p}_j \le p_j^*, \forall j \in [m]$.

        \item \textbf{Step 3: Changes of Utility.} 
        
        From previous steps, we have shown that prices decrease $\hat{p}_j \le p_j^*, \forall j\in [m]$, and $\hat{p}, \hat{x}$ forms a competitive equilibrium according to Lemma \ref{lemma:SOP_IOP}. Combining the fact that the budget $\lambda_{-i}$ and target RoS $\tau_{-i}$ of other buyers expect for buyer $i$ are fixed. Thus $x_{k, \cdot}^*$ are also feasible solutions of Program (\ref{IOP:primal}) when prices are $\hat{p}_j$. Thus, the utility of buyer $k \neq i$ is non-decreasing, i.e., 
        \begin{align*}
            \hat{u}_k = \sum_{j=1}^m v_{k,j}\hat{x}_{k,j} \ge \sum_{j=1}^m v_{k,j}x_{k,j}^* \ge u_k^*, \forall k \neq i, k\in [n].
        \end{align*}
        Combining with Lemma \ref{lemma:Pareto_efficiency}, we have
        \begin{align*}
            \hat{u}_i = \sum_{j=1}^m v_{i,j}\hat{x}_{i,j} \le \sum_{j=1}^m v_{i,j}x_{i,j}^* = u_i^*.
        \end{align*}
        It implies buyer $i$ cannot get higher utility via under-reporting the financial constraints.
   \end{itemize}

   \textbf{Part 2: Over-reporting Constraints, i.e., $\hat{\lambda}_i > \lambda_i, \hat{\tau}_i < \tau_i$.}
   
    According to Lemma \ref{lemma:SOP_IOP}, we know that one of the following conditions holds.
    \begin{itemize}
        \item \textit{Budget constraint is binding.}
        \begin{align*}
            \sum_{j=1}^m \hat{p}_j \hat{x}_{i,j} = \hat{\lambda}_i > \lambda_i.
        \end{align*}
        \item \textit{RoS constraint is binding.}
        \begin{align*}
            \sum_{j=1}^m v_{i,j}\hat{x}_{i,j} = \hat{\tau}_i \sum_{j=1}^m \hat{p}_j \hat{x}_{i,j} < \tau_i \sum_{j=1}^m \hat{p}_j \hat{x}_{i,j}.
        \end{align*}
    \end{itemize}
    In either cases, if the misreporting alters the allocation, it results in an outcome that violates buyer $i$'s true constraints. By definition, the utility of an infeasible outcome is $-\infty$. Thus, buyer $i$ cannot gain by over-reporting.

    \textbf{Conclusion:} Since neither under-reporting nor over-reporting increases utility, truthfully reporting financial constraints is a dominant strategy.

    The proof is completed.
\end{proof}

\subsubsection{Proof of Theorem \ref{thm:revenue_approximation}}
\begin{proof}
We use $\text{Rev}^{FB}$ as an intermediary and show:
\begin{gather*}
    \frac{\text{Rev}^*}{\text{Rev}^{FB}}\ge \frac{1}{2}.
\end{gather*}
And the stated result immediately follows as $\text{Rev}^{FB}$ is an upper bound of $\overline{\text{Rev}}$.
% We calculate the First-Best revenue using the following primal program, where $t_i$ represents the total payment from buyer $i$.
% \begin{maxi}|l|[2]{x,t}{\sum_{i=1}^n t_i}{\label{prog:fb_primal}}{}
% \addConstraint{t_i}{\le \lambda_i}{\quad\forall i \in [n]}
% \addConstraint{t_i}{\le \frac{1}{\tau_i} \sum_{j=1}^m v_{i,j} x_{i,j}}{\quad\forall i \in [n]}
% \addConstraint{\sum_{i=1}^n x_{i,j}}{\le 1}{\quad \forall j \in [m]}
% \addConstraint{x_{i,j}}{\ge 0}{\quad\forall i \in [n], \forall j \in [m]}.
% \end{maxi}

% The corresponding dual program is:

% \begin{mini}|l|[2]{\alpha, \beta, p}{\sum_{i=1}^n \alpha_i \lambda_i + \sum_{j=1}^m p_j}{\label{prog:fb_dual}}{}
% \addConstraint{\alpha_i + \tau_i \beta_i}{\ge 1}{\quad \forall i \in [n]}
% \addConstraint{p_j}{\ge \beta_i v_{i, j}}{\quad \forall i \in [n], j \in [m]} 
% \addConstraint{\alpha_i}{\ge 0}{\quad \forall i \in [n]}
% \addConstraint{\beta_i}{\ge 0}{\quad \forall i \in [n]}
% \end{mini}

We first show that in any optimal solution $(\beta,\alpha,p)$ to Program \eqref{prog:fb_dual}, it must be that $0<\beta_i\le \frac{1}{\tau_i},\forall i$. Otherwise, if $\beta_i=0$ for some buyer $i$, then this buyer cannot win any item. By increasing $\beta_i$ to a sufficiently small positive number, we can lower $\alpha_i$ without affecting the allocation and $p_j,\forall j$, thus decreasing the objective. And if there exists $i$ with $\beta_i>\frac{1}{\tau_i}$, we must have $\alpha_i=0,\forall i$ and $p_j=\max\{\beta_jv_{i,j}\},\forall j$ in order to minimize the objective. In this case, we can further decrease the objective by changing $\beta_i$ to $\frac{1}{\tau_i}$ to decrease $p_j$ without affecting $\alpha_i$. It follows that $\beta^{FB}_i>0,\forall i$ since $\beta^{FB}$ is optimal.

We claim that $\mathcal{B}^*\subseteq \mathcal{B}^{FB}$. To prove the claim, we categorize the buyers into two groups: 
    \begin{gather*}
        G_1=\left\{i\,\middle|\, \beta_i^{FB}\le w^*_i\right\}\quad \text{and} \quad G_2=\left\{i\,\middle|\, \beta_i^{FB}> w^*_i\right\}.
    \end{gather*}
According to the definition of $\mathcal{R}^*$, any buyer $i$ in $\mathcal{R}^*$ must have $w^*_i=\frac{1}{\tau_i}$. It follows that any buyer in $G_2$ cannot be in $\mathcal{R}^*$ as $\beta^{FB}_i$ must satisfy $\beta^{FB}_i\le \frac{1}{\tau_i}$. Therefore, buyers in $G_2$ must be all budget-constrained buyers under the market clearing mechanism, and all RoS-constrained buyers are in $G_1$. This indicates that buyers in $G_2$ should obtain more items under $x^{FB}$ than under $x^*$ since both these allocations follow first-price auction rules, i.e.,
    \begin{gather}
        \sum_{i\in G_2}x^{FB}_{i,j}\ge \sum_{i\in G_2}x^{*}_{i,j}, \forall j \label{eq:additional item} .
    \end{gather}
If there exists buyer $i\in \mathcal{B}^*\cap \mathcal{R}^{FB}$, then by definition, we have $t^{FB}_i<\lambda_i$. Since buyers in $G_2$ win more items under $x^{FB}$ than under $x^*$, we can modify $x^{FB}$ by using the same allocation as in $x^*$ for buyers in $G_2$ and giving the extra items to anyone. In this case, we can charge all buyers in $\mathcal{B}^*$ their full budgets without violating their budget and RoS constraints. This clearly leads to a higher revenue there is already a buyer $i$ with $t^{FB}_i<\lambda_i$, contradicting to the optimality of $x^{FB}$ and $t^{FB}$. Therefore, we must have $\mathcal{B}^*\cap \mathcal{R}^{FB}=\emptyset$, i.e., $\mathcal{B}^*\subseteq \mathcal{B}^{FB}$.

% The complementary slackness conditions for the optimal solutions $(x^{FB}, t^{FB})$ and $(\alpha^{FB}, \beta^{FB}, p^{FB})$ are:
% \begin{gather}
%     x_{i, j}^{FB}\left(p_j^{FB} - \beta_i^{FB} v_{i, j} \right) = 0 \label{eq: kkt_x} \\
%     \alpha_i^{FB} \left(t_i^{FB} - \lambda_i \right) = 0 \label{eq:kkt_alpha} \\
%     \beta_i^{FB} \left(t_i^{FB} - \frac{1}{\tau_i}\sum_{j=1}^m v_{i, j}x_{i,j}^{FB} \right) = 0 \label{eq:kkt_beta} \\
%     p_j^{FB} \left(1 - \sum_{i=1}^n x_{i,j}^{FB} \right) =  0 \label{eq:kkt_p}
% \end{gather}

    For convenience, we define
    \begin{gather*}
        \mathcal{J}_{\mathcal{B}^*}^{FB} = \left\{j \big| j\in [m], x_{i,j}^{FB} > 0, \forall i \in \mathcal{B}^* \right\},\\
        \mathcal{J}_{\mathcal{R}^*}^{FB} = \left\{j \big| j\in [m], x_{i,j}^{FB} > 0, \forall i \in \mathcal{R}^* \right\}
    \end{gather*}
    to be the set of items allocated to buyers in $\mathcal{B}^*$ under the First Best solution and the set of items allocated to buyers in $\mathcal{R}^*$ under the First Best solution, respectively.
    
    Since $\mathcal{B}^* \cup \mathcal{R}^* = [n]$, meaning $\mathcal{J}_{\mathcal{B}^*}^{FB} \cup \mathcal{J}_{\mathcal{R}^*}^{FB} = [m]$.

    Similarly, define
    \begin{gather*}
        \mathcal{J}_{\mathcal{B}^*} = \left\{j \big| j\in [m], x_{i,j}^{*} > 0, \forall i \in \mathcal{B}^* \right\},\\
        \mathcal{J}_{\mathcal{R}^*} = \left\{j \big| j\in [m], x_{i,j}^{*} > 0, \forall i \in \mathcal{R}^* \right\}
    \end{gather*}
    to be the set of items allocated to buyers in $\mathcal{B}$ under the market-clearing mechanism and the set of items allocated to buyers in $\mathcal{R}$ under the market-clearing mechanism, respectively.

    We decompose the first-best revenue:
    \begin{align} 
        \begin{aligned}
        \text{Rev}^{FB} &= \sum_{i \in \mathcal{B}^*} t_i^{FB} + \sum_{i \in \mathcal{R}^*} t_i^{FB} \\
        &=\sum_{i\in\mathcal{B}^*} \lambda_i + \sum_{i\in \mathcal{R}^*} t_i^{FB}.
        \end{aligned}
    \end{align}
    The second term (revenue from $\mathcal{R}$) can be expanded by splitting items based on whether they belong to $\mathcal{J}_{\mathcal{R}}$ or not:
    \begin{align*}
        \sum_{i \in \mathcal{R}^*} t_i^{FB} = \sum_{j \in \mathcal{J}_{\mathcal{R}^*}} \sum_{i\in \mathcal{R}^*} \frac{v_{i,j}}{\tau_i} x_{i,j}^{FB} + \sum_{j \in \mathcal{J}_{\mathcal{R}^*}^{FB} \backslash \mathcal{J}_{\mathcal{R}^*}} \sum_{i\in \mathcal{R}^*} \frac{v_{i,j}}{\tau_i} x_{i,j}^{FB}.
    \end{align*}

    Under the market-clear mechanism, for any item $j \in \mathcal{J}_{\mathcal{R}^*}$, if it is allocated to buyer $i \in \mathcal{R}^*$, it must have
    \begin{gather*}
        \frac{v_{i,j}}{\tau_i} = p_j^*.
    \end{gather*}
    Thus, the first term can be bounded by
    \begin{gather*}
        \sum_{j \in \mathcal{J}_{\mathcal{R}^*}} \sum_{i\in \mathcal{R}^*} \frac{v_{i,j}}{\tau_i} x_{i,j}^{FB} \le \sum_{j \in \mathcal{J}_{\mathcal{R}^*}} p_j^*.
    \end{gather*}
    The second term means that under the market-clearing mechanism, the item is allocated to buyers in $\mathcal{B}^*$, while under the First-Best solution, the item is allocated to buyers in $\mathcal{R}^*$. It implies that for item $j \in \mathcal{J}_{\mathcal{R}^*}^{FB} \backslash \mathcal{J}_{\mathcal{R}^*}$
    \begin{align*}
        p_j^* > \frac{v_{i,j}}{\tau_i}, \forall i \in \mathcal{R}^*,
    \end{align*}
    Otherwise, this item must be allocated to buyers in $\mathcal{R}^*$ under the market-clearing mechanism. It follows that
    \begin{align*}
        \sum_{j \in \mathcal{J}_{\mathcal{R}^*}^{FB} \backslash \mathcal{J}_{\mathcal{R}^*}} \sum_{i\in \mathcal{R}^*} \frac{v_{i,j}}{\tau_i} x_{i,j}^{FB} \le \sum_{j \in \mathcal{J}_{\mathcal{R}^*}^{FB} \backslash \mathcal{J}_{\mathcal{R}^*}} p_j^*.
    \end{align*}
    Note that for item $j \in \mathcal{J}_{\mathcal{R}^*}^{FB} \backslash \mathcal{J}_{\mathcal{R}^*}$, it is allocated to buyers in $\mathcal{B}^*$, meaning that
    \begin{align*}
        \sum_{j \in \mathcal{J}_{\mathcal{R}^*}^{FB} \backslash \mathcal{J}_{\mathcal{R}^*}} p_j^* \le \sum_{i\in \mathcal{B}^*} \lambda_i \le \text{Rev}^*.
    \end{align*}
    Under the market-clearing mechanism, we have
    \begin{align*}
        \text{Rev}^* &= \sum_{j\in \mathcal{J}_{\mathcal{B}^*}}p_j^* + \sum_{j\in \mathcal{J}_{\mathcal{R}^*}} p_j^* \\
        &= \sum_{i \in \mathcal{B}^*} \lambda_i + \sum_{j \in \mathcal{J}_{\mathcal{R}^*}} p_j^*.
    \end{align*}
    Combining these inequalities above, we immediately have
    \begin{align*}
        \text{Rev}^{FB} \le \sum_{i\in\mathcal{B}^*} \lambda_i + \sum_{j \in \mathcal{J}_{\mathcal{R}^*}} p_j^* + \text{Rev}^* = 2 \text{Rev}^*.
    \end{align*}
    
\end{proof}

\subsection{Omitted Proofs in Section \ref{sec:computation}}

\subsubsection{Proof of Theorem \ref{thm:auxiliary_regret}}

Before we analyze the bounds for the above regret, we further define the following auxiliary functions, which are used in our analysis:
\begin{definition}[Auxiliary Functions]\label{def:auxiliary_function}
    For any auction $j$, define $\mathcal{Z}^j: \mathbb{R}_{+}^n \times W \mapsto \mathbb{R}$ as 
    \begin{equation*}
        \mathcal{Z}^j(s, \omega) = \left \langle s, \omega - \omega^1\right\rangle  - \sum_{i=1}^n j \rho_i \log(\omega_i),
    \end{equation*}
    where $s\in \mathbb{R}^n$ is an $n$-dimensional vector, and $\langle \cdot,\cdot\rangle$ denotes the inner product.
    
    Then define the auxiliary value function $V^j(s)$ and the minimizer $\Omega^j(s)$ as 
    \begin{gather*}
    V^j(s) = \inf_{\omega \in W} \left\{\mathcal{Z}^j(s,\omega) \right \},\quad
    \Omega^j(s) = \arginf_{\omega \in W}\left\{\mathcal{Z}^j(s,\omega) \right \}.
	\end{gather*}
\end{definition}

The auxiliary value function $V^j(s)$ is concave in $s$ and its gradient is
\begin{gather}\label{eq:v i}
\nabla V^j(s) = \Omega^j(s)-\omega^1.
\end{gather}
It is Lipschitz continuous due to the strong convexity of the regularizer.
It is easy to see that the dimensions of $\mathcal{Z}^j(s,\omega)$ are independent. Thus $\mathcal{Z}^j$ decomposes across buyers:
\begin{gather}
\mathcal{Z}^j(s,\omega)=\sum_{i=1}^n\mathcal{Z}^j_i(s_i,\omega_i)=\sum_{i=1}^n\left[s_i(\omega_i,\omega^1_i)-j \rho_i \log(\omega_i)\right],\nonumber
\end{gather}

The minimizer $\Omega_i^j(s_i)$ has closed form:
\begin{align}
	\Omega^j_i(s_i)=\arginf_{\omega_i \in W_i}\left\{\mathcal{Z}^j_i(s_i, \omega_i)\right\}=
    \begin{cases}
        \underline{w}_i& \text{if } s_i\ge \frac{j\rho_i}{\underline{w}_i}\\
        \frac{j\rho_i}{s_i}&\text{if } j\rho_i\tau_i<s_i\le \frac{j\rho_i}{\underline{w}_i}\\
        \frac{1}{\tau_i}&\text{if } s_i< j\rho_i\tau_i
    \end{cases}.\label{eq:omega i}
\end{align}

The following Lemma \ref{lemma:temp} provides some properties of the above functions.
\begin{lemma}\label{lemma:temp}
% Define 
%     \begin{align*}
%         V_j(s) = \min_{\omega} \left\{\left \langle s, \omega - \omega^0\right\rangle \} - \sum_{i=1}^n j \frac{\lambda_i}{m} \log(\omega_i) \right \},
%     \end{align*}
% and 
%     \begin{align*}
%         \pi_j(s) = \argmin_{\omega}\left\{\left \langle s, \omega - \omega^0\right\rangle \} - \sum_{i=1}^n j \frac{\lambda_i}{m} \log(\omega_i) \right \}.
%     \end{align*}
    The function $V^j(s)$ is concave and has the following gradient:
    \begin{align*}
        \nabla V^j(s) = \Omega^j(s) - \omega^1.
    \end{align*}
    Moreover, the gradient is Lipschitz continuous with constant $\frac{m}{j \sigma }=\frac{1}{j \min_{i} \rho_i \tau_i^2}$, i.e., 
    \begin{align*}
        \left\|\nabla V^j(s)  - \nabla V^j(s') \right\|_2 \le \frac{m}{j \sigma } \| s - s' \|_2, \forall s, s'\in \mathbb{R}^n,
    \end{align*}
    where $\|\cdot\|_2$ is the Euclidean norm.
\end{lemma}

To prove the auxiliary regret, we first split the auxiliary regret in Definition \ref{def:auxiliary_regret} into two terms, as shown in Lemma \ref{lemma:regret_decomp}

\begin{lemma}[Regret Decomposition]\label{lemma:regret_decomp}
The objective regret $R^m_{obj}$ is bounded by the following two term:
    \begin{equation}
    R^m_{obj} \le \sum_{j=1}^m \left(\left \langle g^j(\omega^j),\omega^j - \omega^1\right \rangle -\sum_{i=1}^n \rho_i\log(\omega^j_i) \right) - V^{m}(m\bar{g}^m).
    \end{equation}
\end{lemma}

\begin{lemma}\label{lemma:recursive_bound}
For any $j \ge 2$, the single-step progress satisfies:
\begin{equation}\label{eq:regret_bound_1}
\langle g^j(\omega^j), \omega^j - \omega^1 \rangle - \sum_{i=1}^n \rho_i \log(\omega_i^{j+1}) \le V^j(j \bar{g}^j) - V^{j-1}((j-1)\bar{g}^{j-1}) + \frac{m}{2\sigma j} \left \|g^j(\omega^j)\right\|_2^2.
\end{equation}
\end{lemma}

\begin{proof}
Combining Lemma \ref{lemma:regret_decomp} and Lemma \ref{lemma:recursive_bound}, we bound the regret by telescoping the sum.

First, apply Lemma \ref{lemma:regret_decomp}:
\begin{align*}
    R^m_{obj} \le \sum_{j=1}^m \left(\langle g^j, \omega^j - \omega^1 \rangle - \sum_{i=1}^n \rho_i \log(\omega_i^j) \right) - V^m(m \bar{g}^m) .
\end{align*}
We shift the index of the log term in the summation to align with Lemma \ref{lemma:recursive_bound}. By adding and subtracting $\sum_{i=1}^n \rho_i \log(\omega_i^{m+1})$ and noting $\sum_{i=1}^n \rho_i \log(\omega_i^1) = \sum_{i=1}^n \rho_i \log(1/\tau_i)$ is constant:
\begin{align*}
    R^m_{obj} \le \sum_{j=1}^m \left( \langle g^j, \omega^j - \omega^1 \rangle - \sum \rho_i \log(\omega_i^{j+1}) \right) + \sum \rho_i \log(\omega_i^{m+1}) - \sum \rho_i \log(\omega_i^1) - V^m(m\bar{g}^m).
\end{align*}

For $j \ge 2$, Lemma \ref{lemma:recursive_bound} gives
\begin{align}\label{eq:auxiliary_bound_1}
    \langle g^j(\omega^j), \omega^j - \omega^1 \rangle - \sum_{i=1}^n \rho_i \log(\omega_i^{j+1}) \le V^j(j \bar{g}^j) - V^{j-1}((j-1)\bar{g}^{j-1}) + \frac{m}{2\sigma j} \left \|g^j(\omega^j)\right\|_2^2.
\end{align}

For $j=1$, we have:
\begin{align*}
    V^1(\bar{g}^1) &= \min_{\omega\in W} \left\{\langle g^1(\omega^1), \omega-\omega^1)\rangle - \sum_{i=1}^n \rho_i \log(\omega_i) \right\} \\
    &= \langle g^1(\omega^1), \omega^2-\omega^1)\rangle - \sum_{i=1}^n \rho_i \log(\omega_i^2)\\
    &\ge - \sum_{i=1}^n \rho_i \log(\omega_i^2)
\end{align*}
The last inequality follows from the fact that $\omega_i^1 = \frac{1}{\tau_i}, i \in [n]$, which is an upper bound for $\omega_i$. 

For $j=1$, we have a similar bound
\begin{align}\label{eq:auxiliary_bound_2}
        \langle g^1(\omega^1), \omega^1-\omega^1\rangle - \sum_{i=1}^n \rho_i \log(\omega_i^2)\le V^1(\bar{g}^1) 
\end{align}
Combining with inequalities (\ref{eq:auxiliary_bound_1}) and inequality (\ref{eq:auxiliary_bound_2}), and summing these creates a telescoping series:
\begin{align*}
    \sum_{j=1}^m \left( \langle g^j, \omega^j - \omega^1 \rangle - \sum \rho_i \log(\omega_i^{j+1}) \right)\le \left( V^m(m \bar{g}^m) - 0 \right) + \sum_{j=2}^m \frac{m}{2\sigma j} \|g^j\|_2^2.
\end{align*}
Substituting this back into the regret bound, the $V^m(m\bar{g}^m)$ term cancels out:
\begin{align*}
    R^m_{obj} &\le \sum_{j=2}^m \frac{m}{2j \sigma} \|g^j\|_2^2 + \left( \sum_{i=1}^n \rho_i \log(\omega_i^{m+1}) - \sum_{i=1}^n \rho_i \log(\omega_i^1) \right).
\end{align*}
Since $\omega_i^{m+1} \le \frac{1}{\tau_i} = \omega_i^1$, the log difference is negative (or zero). Using the bound $\|g^j\|_2^2 \le \bar{v}^2$ and the harmonic sum $\sum_{j=2}^m \frac{1}{j} \le \log (m)$:
\begin{align*}
    R^m_{obj} \le \frac{m \bar{v}^2}{2\sigma} (\log m).
\end{align*}
With $\sigma = \min_{i\in [n]} \left\{\lambda_i \tau_i^2\right\} = \min_{i\in [n]} \left\{m \rho_i \tau_i^2\right\}$ defined in Lemma \ref{lemma:strong_convex}, we conclude $R^m_{obj} = O(\log m)$.
\end{proof}

\subsubsection{Proof of Lemma \ref{lemma:temp}}
\begin{proof}
    We first prove that $V^j(s)$ is concave. For any $\alpha \in [0, 1]$ and any two vectors $s_1$ and $s_2$, we have
    \begin{align*}
        &V^j(\alpha s_1 + (1 - \alpha) s_2) \\
        =& \min_{\omega \in W} \left\{\left\langle \alpha s_1 + (1 - \alpha) s_2, \omega - \omega^1\right \rangle - \sum_{i=1}^n j \rho_i \log(\omega_i) \right\} \\
        \ge &\min_{\omega \in W} \left\{\left\langle \alpha s_1 , \omega - \omega^1\right \rangle - \alpha \sum_{i=1}^n j \rho_i\log(\omega_i) \right\} + \min_{\omega \in W} \left\{\left\langle (1- \alpha) s_2, \omega - \omega^1\right \rangle - (1 - \alpha) \sum_{i=1}^n j \rho_i\log(\omega_i) \right\} \\
        = & \alpha V^j(s_1) + (1- \alpha)V^j(s_2).
    \end{align*}
    This indicates that function $V^j(s)$ is concave. Following the proof of Lemma \ref{lemma:strong_convex}, we know that the function $-\sum_{i=1}^n j \frac{\lambda_i}{m} \log(\omega_i)$
    is strongly convex with convexity parameter $\frac{m}{\sigma j}$. This implies that $V^j(s,\omega)$ is also strongly convex and $\Omega^j(s)$ is unique. 
    % Following the convex analysis, we have
    % \begin{align*}
    %     \Omega^j_i(s) = \begin{cases}
    %         \underline{w}_i & \text{if } \frac{j \lambda_i}{m s_i} \le \underline{w}_i \\
    %         \frac{1}{\tau_i} &\text{if } \frac{j \lambda_i}{m s_i} \ge \frac{1}{\tau_i} \\
    %         \frac{j \lambda_i}{m s_i} &\text{otherwise}.
    %     \end{cases}
    % \end{align*}
    Combining Equation \eqref{eq:v i} and \eqref{eq:omega i}, we have
    \begin{align}
    \label{eq:envelope}
        \frac{\partial V^j(s)}{\partial s_i} =\frac{\mathrm{d} V^j_i(s_i)}{\mathrm{d} s_i}=\frac{\mathrm{d} V^j_i(s_i,\Omega^j_i(s_i))}{\mathrm{d} s_i}= \Omega^j_i(s_i)-w^1_i.
        % \begin{cases}
        %     \underline{w}_i - \omega_i^1 & \text{if } \omega_i = \underline{w}_i\\
        %     \frac{j \lambda_i}{m s_i} - \omega_i^1 & \text{if } \omega_i = \frac{j \lambda_i}{m s_i}\\
        %     \frac{1}{\tau_i} - \omega_i^1 & \text{otherwise}.
        % \end{cases}
    \end{align}
    It follows that\footnote{This result can also be obtained using the well-known envelope theorem.}
    \begin{align*}
        \nabla V^j(s) = \Omega^j(s) - \omega^1.
    \end{align*}
    % According to the envelope theorem, we have that the gradient of $V^j(s)$ is:
    % \begin{gather*}
    %     \nabla V^j(s)=\frac{\partial V^j(s,\omega)}{\partial s} + \frac{\partial V^j(s, \omega)}{\partial \omega}\nabla_s\Omega^j(s) \ge \Omega^j(s)-\omega^1.
    % \end{gather*}

    According to Equation \eqref{eq:envelope}, we know that the $i$-th element of $\nabla V^j(s)-\nabla V^j(s')$ is simply $\Omega^j_i(s_i)-\Omega^j_i(s'_i)$. Using Equation \eqref{eq:omega i}, it is easy to check that $\Omega^j_i(s_i)$ is $\frac{1}{j\rho_i \tau_i^2}$-Lipschitz. This immediately implies that $\nabla V^j(s)$ is $\frac{m}{\sigma j}$-Lipschitz, since
    \begin{gather*}
        \|\nabla V^j(s)-\nabla V^j(s')\|_2^2=\sum_{i=1}^n\left[\Omega^j_i(s_i)-\Omega^j_i(s'_i)\right]^2\le \sum_{i=1}^n\frac{1}{j\rho_i\tau_i^2}\|s-s'\|_2^2\le \sum_{i=1}^n\frac{m}{\sigma j}\|s-s'\|_2^2.
    \end{gather*}
\end{proof}

\subsubsection{Proof of Lemma \ref{lemma:regret_decomp}}
\begin{proof}
By the definition of sub-gradient convexity,
\begin{align*}
    \hat{p}_j(\omega^j) - \hat{p}_j(w^*) \le \langle g^j(\omega^j), \omega^j - w^* \rangle
\end{align*}
We substitute this into the regret definition:
\begin{align*}
R^m_{obj} &= \sum_{j=1}^m \left[\hat{p}_j(\omega^j) - \hat{p}_j(w^*) - \sum_{i=1}^n \rho_i \log(\omega_i^j) + \sum_{i=1}^n \rho_i \log(w_i^*) \right] \\
&\le \sum_{j=1}^m \left( \langle g^j(\omega^j), \omega^j - w^* \rangle - \sum_{i=1}^n \rho_i \log(\omega_i^j) \right) + \sum_{i=1}^n m\rho_i \log(w_i^*) \\
&= \sum_{j=1}^m \left( \langle g^j(\omega^j), \omega^j - \omega^1 + \omega^1 - w^*\rangle - \sum_{i=1}^n \rho_i \log(\omega_i^j) \right)  + \sum_{i=1}^n m \rho_i \log(w_i^*) \\
&= \sum_{j=1}^m \left( \langle g^j(\omega^j), \omega^j - \omega^1 \rangle - \sum_{i=1}^n \rho_i \log(\omega_i^j) \right) - \langle m \bar{g}^m, w^* - \omega^1 \rangle + \sum_{i=1}^n m \rho_i \log(w_i^*).
\end{align*}

The last equality above follows from the definition of function (\ref{eq:dual_average}).
\begin{align*}
    \sum_{j=1}^m \langle g^j(\omega^j), w^*-\omega^1\rangle &= \left \langle \sum_{j=1}^m g^j(\omega^j), w^*-\omega^1\right \rangle = \langle m\bar{g}^j, w^*-\omega^1\rangle
\end{align*}

The term involving $w^*$ is bounded by the minimum over $W$:
\begin{align*}
-\langle m \bar{g}^m, w^* - \omega^1 \rangle + \sum_{i=1}^n m \rho_i \log(w_i^*) &= - \left ( \langle m \bar{g}^m, w^* - \omega^1 \rangle - \sum_{i=1}^n m \rho_i \log(w_i^*) \right) \\
&\le - \min_{w \in W} \left\{ \langle m \bar{g}^m, w - \omega^1 \rangle - \sum_{i=1}^n m \rho_i \log(w_i) \right\} \\
\end{align*}
It follows that
\begin{align}\label{eq:delta}
    R^m_{obj} \le \sum_{j=1}^m \left(\left \langle g^j(\omega^j),\omega^j - \omega^1\right \rangle -\sum_{i=1}^n \rho_i\log(\omega^j_i) \right) - \min_{w \in W} \left\{ \langle m \bar{g}^m, w - \omega^1 \rangle - \sum_{i=1}^n m \rho_i \log(w_i) \right\}
\end{align}
Combining with the definition that
\begin{align*}
    V^{m}(m\bar{g}^m) = \min_{w \in W} \left\{ \langle m \bar{g}^m, w - \omega^1 \rangle - \sum_{i=1}^n m \rho_i \log(w_i) \right\},
\end{align*}
we complete the proof.
\end{proof}

\subsubsection{Proof of Lemma \ref{lemma:recursive_bound}}
\begin{proof}
From the definition of $V^j$ (see Definition \ref{def:auxiliary_function}), we rewrite the term $V^j(j \bar{g}^j)$ by separating the regularizer:
\begin{align*}
&V^j(j\bar{g}^j) + \sum_{i=1}^n \rho_i \log(\omega_i^{j+1})\\
=& \min_{\omega \in W} \left\{ \langle j\bar{g}^j, \omega - \omega^1 \rangle - j\sum_{i=1}^n \rho_i \log(\omega_i) \right\} + \sum_{i=1}^n \rho_i \log(\omega_i^{j+1}) \\
=&\langle j\bar{g}^j, \omega^{j+1} - \omega^1 \rangle - j\sum_{j=1}^m \rho_i \log(\omega_i^{j+1}) + \sum_{i=1}^n \rho_i \log(\omega_i^{j+1})\\
=&\langle j\bar{g}^j, \omega^{j+1} - \omega^1 \rangle - (j-1)\sum_{j=1}^m \rho_i \log(\omega_i^{j+1}) \\
\ge& \min_{\omega \in W} \left\{ \langle j\bar{g}^j, \omega - \omega^1 \rangle - (j-1)\sum \rho_i \log(\omega_i) \right\}.
\end{align*}
Substituting $j\bar{g}^j = (j-1)\bar{g}^{j-1} + g^j(\omega^j)$, it follows that
\begin{align*}
V^j(j\bar{g}^j) + \sum_{i=1}^n \rho_i \log(\omega_i^{j+1})
\ge & \min_{\omega \in W} \left\{ \langle (j-1)\bar{g}^{j-1} + g^j(\omega^j), \omega - \omega^1 \rangle - (j-1)\sum_{i=1}^n \rho_i \log(\omega_i) \right\} \\
=& V^{j-1}\left( (j-1)\bar{g}^{j-1} + g^j(\omega^j) \right).
\end{align*}

From the proof of Lemma \ref{lemma:temp}, we can conclude that 
\begin{align}\label{ineq:V_strong_concave}
    V^j(s_1 + s_2) \ge V^j(s_1) + \left \langle s_2, \nabla V^j(s_1) \right \rangle - \frac{m}{2 j \sigma } \|s_2\|_2^2,\forall s_1,s_2.
\end{align}
Using the inequality \eqref{ineq:V_strong_concave}) and noting that $\nabla V^{j-1}((j-1)\bar{g}^{j-1}) = \omega^j - \omega^1$ (by the Envelope Theorem):
\begin{align*}
V^{j-1}\left( (j-1)\bar{g}^{j-1} + g^j(\omega^j) \right) &\ge V^{j-1}((j-1)\bar{g}^{j-1}) + \langle g^j(\omega^j), \omega^j - \omega^1 \rangle - \frac{m}{2j \sigma } \left\|g^j(\omega^j)\right\|^2_2.
\end{align*}
Rearranging the terms yields the result.
\end{proof}
\subsubsection{Proof of Lemma \ref{lemma:strategy_comparison}}
\begin{proof}
    Recall that the update rule in Algorithm \ref{alg:RDA} is
    \begin{align*}
		\omega^{m+1} = \argmin_{\omega \in W}\left\{\left \langle m\bar{g}^m, \omega - \omega^1 \right \rangle - \sum_{i=1}^n m\rho_i\log(\omega_i) \right\}.
    \end{align*}
    The first-order necessary optimality condition gives:
    \begin{align}\label{eq:strategy_1}
        \left \langle m\bar{g}^m - m\frac{\rho_i}{ \omega^{m+1}}, w^* - \omega^{m+1} \right \rangle \ge 0,
    \end{align}
    % where $\lambda = [\lambda_1, \cdots, \lambda_n]^T$. 
    This means $\langle m\frac{\rho_i}{\omega^{m+1}}, w^* - \omega^{m+1} \rangle\le \langle m\bar{g}^m, w^* - \omega^{m+1} \rangle$.
    The strong convexity of function $-\sum_{i=1}^n \rho_i\log(\omega_i)$ implies
	\begin{align}\label{eq:strategy_2}
		-\sum_{i=1}^n m \rho_i \log(w_i^*) \ge -\sum_{i=1}^n m\rho_i\log(w_i^{m+1}) - m \left\langle \frac{\rho}{\omega^{m+1}}, w^* - \omega^{m+1} \right \rangle + \frac{\sigma}{2}\|w^* - \omega^{m+1} \|_2^2.
	\end{align}
 Rearranging the inequality (\ref{eq:strategy_2}) 
 and combining the inequality (\ref{eq:strategy_1}), it follows that
	\begin{align*}
		&\frac{\sigma}{2}\left\|w^* - \omega^{m+1} \right\|_2^2 \\
        \le& -\sum_{i=1}^n m \rho_i\log(w_i^*)+\sum_{i=1}^n m\rho_i\log(\omega_i^{m+1}) + \left\langle m\frac{\rho}{\omega^{m+1}}, w^* - \omega^{m+1} \right \rangle \\
		\le& -\sum_{i=1}^n m \rho_i \log(w_i^*)+\sum_{i=1}^n m\rho_i\log(\omega_i^{m+1}) + \left \langle m \bar{g}_m,  w^* -\omega^{m+1}\right \rangle \\
		=&-\sum_{i=1}^n m \rho_i\log(w_i^*)+\sum_{i=1}^n m \rho_i\log(\omega_i^{m+1}) + \left \langle  m\bar{g}_m,  w^* + \omega^1 - \omega^{m+1} - \omega^1\right \rangle \\
		=& -\sum_{i=1}^n m \rho_i\log(w_i^*) + \left \langle m \bar{g}_m, w^* - \omega^1 \right \rangle - \min_{\omega \in W}\left\{ \left \langle m \bar{g}_m,  \omega-\omega^1\right \rangle - \sum_{i=1}^n m \rho_i\log(\omega_i)\right \}.
	\end{align*}
    Note that the term $\langle m \bar{g}_m, w^*-\omega^1 \rangle$ can be re-written as:
    \begin{align*}
		&\left \langle m \bar{g}_m, w^*-\omega^1 \right \rangle \\
        =& \sum_{j=1}^m \left \langle g^j(\omega^{j}), w^* - \omega^1 \right \rangle \\
		=& \sum_{j=1}^m \left \langle g^j(\omega^{j}), \omega^j - \omega^1\right \rangle + \sum_{j=1}^m \left \langle g^j(\omega^{j}), w^* - \omega^j\right \rangle- \sum_{j=1}^m \sum_{i=1}^n\rho_i\log(\omega_i^j) + \sum_{j=1}^m \sum_{i=1}^n\rho_i\log(\omega_i^j).
    \end{align*}
    Defining the right-hand side of Equation \eqref{eq:delta} as $\Delta$, we have:
    \begin{align*}
		\frac{\sigma}{2}\left\|w^* - \omega^{m+1} \right\|_2^2 \le& -\sum_{i=1}^n m \rho_i \log(w_i^*) + \sum_{j=1}^m \left \langle g^j(\omega^{j}), \omega^j - \omega^1 \right \rangle -\min_{\omega \in W}\left\{ \left \langle m \bar{g}_m,  \omega-\omega^1\right \rangle - \sum_{i=1}^n m \rho_i\log(\omega_i)\right \} \\
        &+\sum_{j=1}^m \left \langle g^j(\omega^{j}), w^* - \omega^j\right \rangle -\sum_{j=1}^m \sum_{i=1}^n\rho_i\log(\omega_i^j) + \sum_{j=1}^m \sum_{i=1}^n\rho_i\log(\omega_i^j)
        \\
		\le & \Delta + \sum_{j=1}^m \left \langle g^j(\omega^{j}), w^* - \omega^j\right \rangle - \sum_{i=1}^n m\rho_i \log(w_i^*)+\sum_{i=1}^n \sum_{j=1}^m\rho_i \log(\omega_i^j)
\end{align*}
Combining with the definition of sub-gradient convexity,
\begin{align*}
    \hat{p}_j(\omega^j) - \hat{p}_j(w^*) \le \langle g^j(\omega^j), \omega^j - w^* \rangle,
\end{align*}
It follows that
\begin{align*}
    &\sum_{j=1}^m \left \langle g^j(\omega^{j}), w^* - \omega^j\right \rangle - \sum_{i=1}^n m\rho_i \log(w_i^*)+\sum_{i=1}^n \sum_{j=1}^m\rho_i \log(\omega_i^j) \\
    \le &- \left[\sum_{j=1}^m \hat{p}_j(\omega^j) - \sum_{j=1}^m\hat{p}_j(w^*) - \sum_{i=1}^n m\rho_i \log(w_i^*)+\sum_{i=1}^n \sum_{j=1}^m\rho_i \log(\omega_i^j) \right] = - R_{obj}^m
\end{align*}

We then have that
\begin{align*}
    \left\|w^* - \omega^{m+1} \right\|_2^2 \le\frac{\delta}{2}(\Delta - R_{obj}^m) \le \frac{\delta}{2}\Delta \le \frac{m \bar{v}^2}{\sigma^2}\log (m)=O\left(\frac{\log m}{m}\right),
\end{align*}
where the last inequality follows from Theorem \ref{thm:auxiliary_regret}.
It immediately implies
    \begin{align*}
        \left|\omega_i^{m+1} - w_i^*\right| \le O \left(\sqrt{\frac{\log m}{m}} \right)
    \end{align*}

\end{proof}

\subsubsection{Proof of Lemma \ref{lemma:utility_bound}}
\begin{proof}

Fix the value of the winner, buyer 1, as $v_{1, j+1}$. Consider the randomness coming from other buyers' valuations $v_{-1, j+1}$. Assume under optimal multipliers, buyer 1 wins, i.e., $w_1^* v_{1, j+1} > w_i^*v_{i, j+1}$ for all $i \neq 1$.

Recall that on Lemma \ref{lemma:strong_convex}, we know that
\begin{align*}
    \delta = \min_{i\in[n]} \left\{\lambda_i \tau_i^2 \right\} \ge m \underline{\rho}, 
\end{align*}
with the conditions that $\lambda_i = m\rho_i \ge m \underline{\rho}, \tau_i \ge 1,\forall i \in [n]$.

Based on Lemma \ref{lemma:strategy_comparison}, we know that for each buyer, when he uses the multiplier learned by Algorithm \ref{alg:RDA}, the following condition holds.
\begin{align*}
    w_i^* - \bar{v}\sqrt{\frac{ \underline{\rho} \log (j)}{j}} \le \omega_i^{j+1} \le w_i^* + \bar{v}\sqrt{\frac{ \underline{\rho} \log (j)}{j}}.
\end{align*}

We further have that buyer $i$'s bid $\omega_i v_{i, j+1}$ falls in the interval
\begin{gather*}
    w^*_i v_{i, j+1} -  \bar{v}^2\sqrt{\frac{ \underline{\rho} \log (j)}{j}} \le \omega^{j+1}_i v_{i, j+1} \le w_i^* v_{i, j+1}+\bar{v}^2\sqrt{\frac{ \underline{\rho} \log (j)}{j}}
\end{gather*}

Then, define 
\begin{gather*}
    \epsilon_{j} = \bar{v}^2\sqrt{\frac{ \underline{\rho} \log (j)}{j}}
\end{gather*}

Then, under Algorithm \ref{alg:RDA}, false allocation occurs when other buyers' bids overlap buyer $1$'s bid. In detail, for buyer $i \neq 1$,
\begin{gather*}
    w_i^*v_{i,j+1} + \epsilon_j \ge w_1^*v_{1, j+1} - \epsilon_j
\end{gather*}

Also, we know that
\begin{gather*}
    w_i^* v_{i, j+1} \le w_1^*v_{1, j+1}
\end{gather*}

Combining these two inequalities, we immediately have
\begin{gather*}
    \frac{w_1^*}{w_i^*}v_{1, j+1} - 2\frac{\epsilon_j}{w_i^*} \le v_{i,j+1} \le \frac{w_1^*}{w_i^*}v_{1, j+1}
\end{gather*}

Since each buyer's value is drawn from a continuous distribution, with  $f_i(v_i) \le \bar{f},\forall v_i \in \mathcal{Z}$.

Thus, the probability that the above event occurs is
% Check
\begin{gather*}
    \int_{\frac{w_1^*}{w_i^*}v_{1, j+1} - 2\epsilon_j}^{\frac{w_1^*}{w_i^*}v_{1, j+1}}f_i(v_i) d v_i \le 2\bar{f}\epsilon_j
\end{gather*}

Using the union bound over all $n-1$ other buyers:
\begin{gather*}
    \text{Pr(Any False Allocation)} \le \sum_{i\neq 1} \frac{2 \bar{f} \epsilon_j}{w_i^*} \le \frac{2(n-1)}{\underline{w}} \epsilon_j
\end{gather*}

Thus, the utility loss at auction $j+1$ is at most the maximum valuation $\bar{v}$ times the probability of false allocation. 

Summation over all auctions, we have that
\begin{gather*}
    \sum_{j=1}^m \frac{2\bar{v}(n-1) \epsilon_j}{\underline{w}} = \frac{2(n-1)\bar{v}^3\bar{f}}{ \underline{w}}\sum_{j=1}^m \sqrt{\frac{\underline{\rho}\log(j)}{j}}
\end{gather*}
In the following part, we show that
\begin{align*}
    \sum_{j=1}^m \sqrt{\frac{ \log(j)}{j}} &= \sum_{j=2}^m \sqrt{\frac{\log(j)}{j}} \\
    &\le \int_{1}^m \sqrt{\frac{ \log(s)}{s}} \D s \\
    &= 2 \sqrt{m\log(m)} - \int_{1}^m \sqrt{\frac{1}{s \log(s)}} ds \\
    &\le 2 \sqrt{m \log(m)}.
\end{align*}

Thus, we know that each buyer's utility not obtained is bounded by
\begin{align*}
    u_i^* - u_i^{m} \le \frac{4(n-1)\bar{v}^3\bar{f}}{\underline{w}
    }\sqrt{\underline{\rho} m \log(m)} = O\left(\sqrt{m\log(m)}\right)
\end{align*}
\end{proof}

\subsubsection{Proof of Lemma \ref{lemma:seller_revenue}}
\begin{proof}
    Define $k_j, j\in [m]$ as the winner of auction $j$ when all buyers use strategy $w_i^*$.
    Thus, we have
    \begin{align*}
        \sum_{j=1}^m \hat{p_j}(w_i^*) = \sum_{j=1}^m w_{k_j}^* v_{k_j, j}
    \end{align*}
    When in auction $j$, buyer $i \neq k_j$ wins auction under Algorithm \ref{alg:RDA}, meaning that $\omega_i^j v_{i,j} \ge \omega_{k_j}^j v_{k_j,j}$.
    We consider the worst case that for each auction $j$, the following conditions hold
    \begin{align*}
        w_{k_j}^* \ge \omega_{k_j^j} \ge w_{k_j}^*-\bar{v}\sqrt{\frac{\underline{\rho}\log (j)}{j}}
    \end{align*}
    Thus, we have
    \begin{align*}
        \sum_{j=1}^m \hat{p}_j(\omega^j) -\sum_{j=1}^m \hat{p}_j(w^*) &\ge -\bar{v} \sum_{j=1}^m \left|w_{k_j}^* - \omega_{k_j}^j \right| \\
        &=-\bar{v}^2 \sum_{j=1}^m \sqrt{\frac{\underline{\rho}\log(j)}{j}} \\
        &\ge-\bar{v}^2\sqrt{\underline{\rho} m \log m} = -O\left(\sqrt{m \log(m)} \right)
    \end{align*}
    Rearrange we have that
    \begin{align*}
        \sum_{j=1}^m \hat{p}_j(\omega^j) \ge \sum_{j=1}^m \hat{p}_j(w^*) - O\left(\sqrt{m \log(m)} \right)
    \end{align*}
\end{proof}

\end{document}